\documentclass[a4paper]{article}
\usepackage[pdftex]{graphicx}
\usepackage{amsmath,amssymb,amsthm,color}
\newtheorem{theorem}{Theorem}
\newtheorem{lemma}[theorem]{Lemma}

\newtheorem{corollary}[theorem]{Corollary}
\newtheorem{definition}[theorem]{Definition}

\usepackage{hyperref}
\usepackage{subfig}
\usepackage{cite}
\usepackage{caption}

\usepackage{algorithm,algorithmic}
\newcommand{\ot}{\leftarrow}

\renewcommand{\O}{\mathrm{O}}
\renewcommand{\mid}{\,:\,}

\newcommand{\PARTITION}{{\sc Partition}}
\newcommand{\PRODUCTPARTITION}{{\sc ProductPartition}}
\newcommand{\strongly}{}

\renewcommand{\theenumi}{$(\alph{enumi})$}

\usepackage[margin=1.03in]{geometry}
\newcommand{\nop}[1]{}

\newcommand{\siff}{\Leftrightarrow}

\stepcounter{secnumdepth}
\stepcounter{tocdepth}

\title{Optimal Composition Ordering Problems for Piecewise Linear Functions}
\author{Yasushi Kawase\thanks{Tokyo Institute of Technology. E-mail: {\tt kawase.y.ab@m.titech.ac.jp}}
\and Kazuhisa Makino\thanks{Kyoto University. E-mail: {\tt makino@kurims.kyoto-u.ac.jp}}
\and Kento Seimi\thanks{The Toa Reinsurance Company, Limited. E-mail: {\tt kento.seimi@gmail.com}}
}
\begin{document}
\maketitle

\begin{abstract}
In this paper, we introduce maximum composition ordering problems.
The input is $n$ real functions $f_1,\dots,f_n:\mathbb{R}\to\mathbb{R}$
and a constant $c\in\mathbb{R}$.
We consider two settings: total and partial compositions.
The maximum total composition ordering problem is to compute
a permutation $\sigma:[n]\to[n]$ which maximizes $f_{\sigma(n)}\circ f_{\sigma(n-1)}\circ\dots\circ f_{\sigma(1)}(c)$, where $[n]=\{1,\dots,n\}$.
The maximum partial composition ordering problem is to compute
a permutation $\sigma:[n]\to[n]$ and a nonnegative integer $k~(0\le k\le n)$
which maximize $f_{\sigma(k)}\circ f_{\sigma(k-1)}\circ\dots\circ f_{\sigma(1)}(c)$.

We propose $\O(n\log n)$ time algorithms for
the maximum total and partial composition ordering problems for monotone linear functions $f_i$, which generalize linear deterioration and shortening models for the time-dependent scheduling problem.
We also show that the maximum partial composition ordering problem can be solved in polynomial time if $f_i$ is of form
$\max\{a_ix+b_i,c_i\}$ for some constants $a_i\,(\ge 0)$, $b_i$ and $c_i$.
%As a corollary, we show that the two-valued free-order secretary problem can be solved in polynomial time.
We finally prove that there exists no constant-factor approximation algorithm
for the problems, even if $f_i$'s are monotone, piecewise linear functions with at most two pieces,
unless P=NP.
\end{abstract}

\section{Introduction}
In this paper, we introduce {\em optimal composition ordering problems} and mainly study their time complexity.
The input of the problems is \(n\) real functions \(f_1,\dots,f_n:\mathbb{R}\to\mathbb{R}\)
and a constant \(c\in\mathbb{R}\). 
In this paper, we assume that the input functions are piecewise linear,
and the input length of a piecewise linear function is 
the sum of the sizes of junctions and coefficients of linear functions.
We consider two settings:  {\em total} and {\em partial} compositions.
The maximum total composition ordering problem is to compute 
a permutation \(\sigma:[n]\to[n]\) that maximizes \(f_{\sigma(n)}\circ f_{\sigma(n-1)}\circ\dots\circ f_{\sigma(1)}(c)\), 
where \([n]=\{1,\dots,n\}\).
The maximum partial composition ordering problem is to compute
a permutation \(\sigma:[n]\to[n]\) and a nonnegative integer \(k~(0\le k\le n)\) that maximize \(f_{\sigma(k)}\circ f_{\sigma(k-1)}\circ\dots\circ f_{\sigma(1)}(c)\).
For example, 
if the input consists of  $f_1(x)=2x-6$, $f_2(x)=\frac{1}{2}x+2$, $f_3(x)=x+2$, and $c=2$,
then the ordering $\sigma$ such that $\sigma(1)=2$, $\sigma(2)=3$, and $\sigma(3)=1$ is optimal for the maximum total composition ordering problem. 
In fact, $f_1\circ f_3\circ f_2(c)=f_1(f_3(f_2(c)))=f_1(f_3(c/2+2))=f_1(c/2+4)=c+2=4$ provides the optimal value of the problem. 
The ordering $\sigma$ above and $k=2$ is optimal for the maximum partial composition ordering problem, where $f_3\circ f_2(c)=5$.
We remark that the minimization versions are equivalent to the maximization ones. %, which can be shown in the next section.

We also consider the maximum exact $k$-composition ordering problem,
which is a problem to compute 
a permutation \(\sigma:[n]\to[n]\) that maximizes \(f_{\sigma(k)}\circ f_{\sigma(k-1)}\circ\dots\circ f_{\sigma(1)}(c)\) for given $n$ functions $f_1,\dots,f_n:\mathbb{R}\to\mathbb{R}$,
a constant \(c\in\mathbb{R}\), and a nonnegative integer \(k~(0\le k\le n)\).

As we will see in this paper, the optimal composition ordering problems are natural and fundamental in many fields such as artificial intelligence, computer science, and operations research. 
However, to the best of the authors' knowledge, no one explicitly studies the problems from the algorithmic point of view. 
We below describe the single machine time-dependent scheduling problems 
and the free-order secretary problem, which can be formulated as the optimal composition ordering problems. 
%% We also show that the subset sum problem and the Hamilton cycle problem are 
%% naturally reducible to the optimal composition ordering problems in Appendix.

\subsection*{Time-dependent scheduling}
Consider the machine scheduling problems with time-dependent processing times,
called {\em time-dependent scheduling problems} \cite{cheng2004acs,gawiejnowicz2008tds}. 

Let  \(J_i\) \((i=1,\dots,n)\)  denote  a job with
a ready time \(r_i \in \mathbb{R}\), a deadline \(d_i \in \mathbb{R}\), 
and a processing time \(p_i:\mathbb{R}\to\mathbb{R}\), where $r_i \leq d_i$ is assumed. 
Different from the classical setting,
the processing time \(p_i\) is \emph{not} constant,
but depends on the \emph{starting} time of job \(J_i\).
The model has been studied to deal with learning and deteriorating effects, 
for example \cite{gawiejnowicz1995sjw,gupta1988sfs,ho1993cos,tanaev1994sts,wajs1986paf}.
Here each \(p_i\) is assumed to satisfy \(p_i(t)\le s+p_i(t+s)\) for any \(t\) and \(s\ge 0\),
since we should be able to finish processing job \(J_i\) earlier if it starts earlier.
%we can earlier finish processing the job \(J_i\) if it is earlier started processing.
Among time-dependent settings, we consider the single machine scheduling problem to minimize the makespan, 
where the input is  the start time \(t_0~(=0)\) and a set of   \(J_i\) \((i=1,\dots,n)\) above. 
The makespan denotes the time when all the jobs have finished processing, and 
we assume that the machine can handle only one job at a time and preemption is not allowed.
We show  that the problem can be seen as the minimum total composition ordering problem.

For simplicity, let us first consider the simplest case, that is,  
each job has neither the ready time \(r_i\) nor the deadline \(d_i\).
Let \(c=t_0\), and 
for each $i\in[n]$, define the function $f_i$ by \(f_i(t)=t+p_i(t)\). 
Note that job \(J_i\) has been finished processing at time \(f_i(t)\)
if it is started processing at time \(t\).
This implies that \(f_{\sigma(n)}\circ f_{\sigma(n-1)}\circ\dots\circ f_{\sigma(1)}(t_0)\)
denotes the makespan of the scheduling problem 
when we fix the ordering \(\sigma\) of the jobs.
Therefore, the problem is represented as the minimum total composition ordering problem. 
More generally, let us consider the case in which each job \(J_i\) also has 
both the ready time \(r_i\) and the deadline \(d_i\)  with \(d_i\ge r_i\).
Define the function $f_i$ by 
\begin{align*}
f_i(t)=
\begin{cases}
r_i+p_i(r_i) &(t\le r_i),\\
t+p_i(t)     &(r_i<t\le d_i-p_i(t)),\\
\infty       &(t>d_i-p_i(t)).
\end{cases}
\end{align*}
Then the problem can be reduced to the minimum total composition ordering problem \(((f_i)_{i\in [n]},c=t_0)\). 
A number of restrictions on the processing time $p_i(t)$ has been studied in this literature (e.g., \cite{cai1998oas,cheng2003sjw,melnikov1980ppo}). 
%We here describe the linear deterioration and shortening models.

In the {\em linear deterioration} model, 
the processing time $p_i$ is restricted to be 
a monotone increasing linear function that satisfies \(p_i(t)=a_i t+b_i\) for two positive constants $a_i$ and $b_i$.
Here \(a_i\) and \(b_i\) are respectively called
the {\em deterioration rate} and the {\em basic processing time}
of job \(J_i\).
Gawiejnowicz and Pankowska \cite{gawiejnowicz1995sjw},
Gupta and Gupta \cite{gupta1988sfs},
%Browne and Yechiali \cite{browne1990sdj},
Tanaev {\it et al.} \cite{tanaev1994sts},
and Wajs \cite{wajs1986paf}
obtained the result that 
the time-dependent scheduling problem of this model (without the ready time $r_i$ nor the deadline $d_i$)  is solvable in \(\O(n\log n)\) time 
by scheduling jobs in the nonincreasing order of ratios \(b_i/a_i\).
% Gawiejnowicz and Pankowska \cite{gawiejnowicz1995sjw}  extended the result to 
%% the case in which \(a_i=0\) or \(b_i=0\). 
%% Monsheiov \cite{mosheiov1994sju}
%% considered the {\em proportional deterioration} model, i.e., \(b_i=0\) for all $i\in[n]$,
%% and showed that the makespan  does not depend on the job processing ordering, i.e., 
%% any scheduling provides an optimal makespan.
%% Cheng and Ding \cite{cheng2000sms} provided
%% an \(\O(n^5)\)-time algorithm for the model 
%% \(p_i(t)=at+b_i\) \((a,b_i>0)\) with a deadline \(d_i\), where we note that all $p_i$'s have the same slope $a$. 
%% For the {\em nonlinear} deterioration model,  
%% it is known that the problem can be solved in \(\O(n\log n)\) time, if $p_i$ satisfies $p_i(t)=g(t)+b_i$ for some nondecreasing (general)  function $g$ \cite{melnikov1980ppo}.
As for the hardness results, it is known that 
the proportional deterioration model with ready time and deadline, 
the linear deterioration model with ready time, 
and the linear deterioration model with a deadline are all \strongly NP-hard \cite{cheng1998tco,gawiejnowicz2007sdj}.

Another important model is called the {\em linear shortening} model introduced by Ho {\it et al.} \cite{ho1993cos}.
In this model, the processing time $p_i$ is restricted to be a monotone decreasing linear function
that satisfies \(p_i(t)=-a_i t+b_i\) with two constants $a_i$ and $b_i$ with $1>a_i>0$, $b_i>0$.
%, and  \(a_j\left(\sum_{i=1}^n b_i-b_j\right)<b_j\).
%% Here the assumptions on $a_i$ and $b_i$ make sense from the practical point of view
%% (e.g., the processing time is nonnegative). %% add meaning
They showed
that the time-dependent scheduling problem of this model can be solved in \(\O(n\log n)\) time 
by again scheduling jobs in the nonincreasing order of the ratios \(b_i/a_i\).

\subsection*{Free-order secretary problem}
The {\em free-order secretary problem} is another application of the optimal composition ordering problems, 
which is closely related to 
a branch of the problems such as the full-information secretary problem \cite{ferguson1989wst},
knapsack and matroid secretary problems \cite{babaioff2007aks,babaioff2007msp,gharan2011ovo}
and stochastic knapsack problems \cite{dean2005aaa,dean2008ats}. 
Imagine that an administrator wants to hire the best secretary out of \(n\) applicants for a position.
Each applicant \(i\) has a nonnegative independent random variable \(X_i\) as his ability for the secretary.
Here \(X_1,\dots,X_n\) are not necessarily based on the same probability distribution,
and assume that the administrator knows all the probability distributions of $X_i$'s before their interviews, 
where such information can be obtained by their curriculum vitae and/or results of some written examinations. 
The applicants are interviewed one-by-one, and 
the administrator can observe the value \(X_i\) during the interview of the applicant $i$. 
A decision on each applicant is to be made immediately after the interview.
Once an applicant is rejected, he will never be hired.
The interview process is finished if some applicant is chosen, where we assume that 
the last applicant is always chosen if he is interviewed
since the administrator has to hire exactly one candidate. 
The objective is to find an optimal strategy for this interview process,
i.e., to find an interview ordering together with the stopping rule that maximizes the expected value of the secretary hired.

Let \(f_i(x)=\mathbf{E}[\max\{X_i,x\}]\).
For example, let us assume that \(X_i\) is an $m$-valued random variable 
that takes the value \(a_i^j\) with probability \(p_i^j \geq 0\) (\(j=1,\dots,m\)).
Here we assume that \(a_i^1\ge\dots\ge a_i^{m} \ge 0\) and  $\sum_{j=1}^m p_i^j=1$. 
Then we have 
\begin{align*}
f_i(x)
&=\sum_{j=1}^m p_i^j\max\{a_i^j,x\}%\\
%% &=\begin{cases}
%% \sum_{j=1}^mp_i^ja_i^j&   \text{if}~~ x \leq a_i^m, \\
%% \sum_{j=1}^{l}p_i^ja_i^j+\sum_{j=l+1}^{m}p_i^j x & \text{if}~~ a_i^{l+1} < x \leq a_i^{l} \ \ \  \left(l=1,\dots, m-1 \right),\\ 
%% x &  \text{if}~~ x> a_i^1,
%% \end{cases}\\
=\max_{l=0, \dots , m}\left\{\sum_{j=1}^{l}p_i^ja_i^j+\sum_{j=l+1}^{m}p_i^j x\right\}.
\end{align*}
Note that this $f_i$ is a  monotone piecewise linear function with at most \((m+1)\) pieces. 
We now claim that our secretary problem can be represented by the maximum total composition ordering problem \(((f_i)_{i\in[n]},c=0)\).

Let us consider  the best stopping rule for the interview to maximize the expected value  for the secretary hired
when the interview ordering is fixed in advance. 
 Assume that the applicant $i$ is interviewed in the $i$th place.
Note that \(\mathbf{E}[X_n]\,(=f_{n}(0))\) is the expected value  under the condition that 
all the applicants except for the last one are rejected, since the last applicant is hired.   
Consider the situation that all the applicants except for the last two ones are rejected. 
Then it is a best stopping rule that 
 the applicant $n-1$ is hired if and only if $X_{n-1} \geq f_{n}(0)$ is satisfied (i.e., 
 the applicant $n$ is hired if and only if $X_{n-1} < f_{n}(0)$), 
where $f_{n-1}\circ f_n(0)$ is  the expected value for the best stopping rule,   under this situation. 
By applying backward induction,  we have the following best stopping rule:
we hire the applicant $i\,(< n)$  and stop the interview process,  
if  $X_i \geq f_{i+1}\circ\dots\circ f_{n}(0)$  (otherwise, the next applicant is interviewed), 
and we hire the applicant $n$ if no applicant $i \,(< n)$ is hired. 
We show that  
$f_{1}\circ\dots\circ f_{n}(0)$ is the maximum expected value for the secretary hired, if the  
interview ordering is fixed such that the applicant $i$ is interviewed in the $i$th place.

Therefore, the secretary problem (i.e., finding an interview ordering, together with a stopping rule) can be formulated as the maximum total composition ordering problem \(((f_i)_{i\in[n]},c=0)\). 
%We note that it can also be seen as the partial composition ordering problem from the discussion in Appendix. 

\subsection*{Main results obtained in this paper}
In this paper, we consider the computational issues for the optimal composition ordering problems,  when all $f_i$'s are monotone and almost linear. 

We first show that the problems become tractable if all $f_i$'s are monotone and linear, i.e., $f_i(x)=a_ix+b_i$ for $a_i \geq 0$. 
%% Theorem 1,2
\begin{theorem}\label{theorem:partial-total}
The maximum partial and total composition ordering problems for monotone nondecreasing linear functions
are both solvable  in \(\O(n\log n)\) time.
\end{theorem}
Recall that the algorithm for the linear shortening model (resp., the linear deterioration model) for the time-dependent scheduling problem
is easily generalized to the case when all $a_i$'s satisfy $a_i< 1$ (resp., $a_i > 1$). 
The best composition ordering is obtained as the nondecreasing order of ratios \(b_i/a_i\).
% (resp., the nondecreasing order of ratios \(b_i/a_i\)). 
This idea can be extended to the maximum partial composition ordering problem in the mixed case (i.e., some $a_i>1$ and some $a_{i'}<1$) of Theorem \ref{theorem:partial-total}.
However, we cannot extend it to the maximum total composition ordering problem.  
In fact, we do not know if there exists such a simple criterion on the maximum total composition ordering.  
We instead present an efficient algorithm that chooses the best ordering among linearly many candidates.

We also provide a dynamic-programming based polynomial-time algorithm for the exact $k$-composition setting.
\begin{theorem}\label{theorem:exact}
The maximum exact $k$-composition ordering problem for monotone nondecreasing linear functions is solvable in \(\O(k\cdot n^2)\) time.
\end{theorem}

We next consider monotone, piecewise linear case.
It can be directly shown from the time-dependent scheduling problem that the maximum total composition ordering problem is NP-hard, 
even if all $f_i$'s are monotone, concave, and piecewise linear functions with at most two pieces, i.e., 
$f_i(x)=\min\{a_i^1x+b_i^1,\,a_i^2x+b_i^2\}$ for some constants $a_i^1$, $a_i^2$, $b_i^1$, 
and $b_i^2$ with $a_i^1,a_i^2>0$.
It turns out that all the other cases become intractable, even if all $f_i$'s are monotone and consist of at most two pieces.
Furthermore, the problems are inapproximable.

\begin{theorem}\label{theorem:hardness}
\noindent
{\rm (i)}  %The maximum total composition ordering problem is NP-hard, even if 
% all $f_i$'s are monotone, concave, and piecewise linear functions with at most two pieces. 
For any positive real number \(\alpha~(\le 1)\),
there exists no $\alpha$-approximation algorithm 
for the maximum total (partial) composition ordering problem
even if all $f_i$'s are monotone, concave, and piecewise linear functions with at most two pieces,
unless P=NP.

\noindent
{\rm (ii)} %The maximum total and partial composition ordering problems are both NP-hard, even if 
% all $f_i$'s are monotone, convex, and piecewise linear functions with at most two pieces.
For any positive real number \(\alpha~(\le 1)\),
there exists no $\alpha$-approximation algorithm 
for the maximum total (partial) composition ordering problem
even if all $f_i$'s are monotone, convex, and piecewise linear functions with at most two pieces,
unless P=NP.
\end{theorem}
\noindent
Note that $f_i$ can be represented by $f_i(x)=\max\{a_i^1x+b_i^1,\,a_i^2x+b_i^2\}$ for some constants $a_i^1$, $a_i^2$, $b_i^1$, 
and $b_i^2$ with $a_i^1,a_i^2>0$ if $f_i$ is a monotone, convex, and piecewise linear function with at most two pieces. 

As for the positive side, if each $f_i$ is a monotone, convex, and piecewise linear function with at most two pieces such that one of the pieces is constant, then 
we have the following result, which implies that  
the  two-valued free-order secretary problem can be solved in 
 \(\O(n^2)\) time.
%% Theorem 3
\begin{theorem}\label{theorem:pconst}
Let $f_i(x)=\max\{a_ix+b_i,c_i\}$ for some constants $a_i\,(\ge 0)$, $b_i$ and $c_i$.
Then the maximum partial composition ordering problem is solvable in \(\O(n^2)\) time.
\end{theorem}

We summarize the current status on 
the time complexity of the maximum total composition ordering problem
in Table~\ref{table:result-maxtcop}.
Here the bold letters represent our results,
and the results for the minimum and/or partial versions are described as the ones 
for the maximum total composition ordering problem, 
since the minimum  and  partial versions 
can be transformed into the maximum total one as shown in Section \ref{sec:maxpcop}.

\begin{table}[hb]  
\caption{The current status on the time complexity of the maximum total composition ordering problem.}
\label{table:result-maxtcop}
\def\tabcolsep{5pt}
\def\arraystretch{1.2}
\begin{center}
\begin{tabular}{lcc}
\hline Functions&Complexity&References\\\hline
$f_i(x)=a_ix$ ($a_i> 1$)&$\O(n)$&\cite{mosheiov1994sju}\\
$f_i(x)=a_ix+b_i$\quad($a_i>1,~b_i<0$)&$\O(n\log n)$&{\small\cite{gawiejnowicz1995sjw,gupta1988sfs,tanaev1994sts,wajs1986paf}}\\
%$f_i(x)=\begin{cases}ax+b_i&(x\ge d_i)\\-\infty&(x<d_i)\end{cases}$\quad($a>1,\ b_i<0$)&$\O(n^5)$&\cite{cheng2000sms}\\
%$f_i(x)=x+g(x)+b_i$ ($b_i<0$, $g(x)$: nondecreasing)&$\O(n\log n)$&\cite{melnikov1980ppo}\\
$f_i(x)=\min\{ax+b_i,r_i\}$\quad($a>1,~b_i<0$)&NP-hard&\cite{cheng1998tco}\\
$f_i(x)=\begin{cases}\min\{a_ix,r_i\}&(x\ge d_i)\\-\infty&(x<d_i)\end{cases}$\quad($a_i>1$)&\strongly NP-hard&\cite{gawiejnowicz2007sdj}\\
$f_i(x)=\min\{a_ix+b_i,c_i\}$\quad $(a_i>1)$&\strongly NP-hard&\cite{cheng1998tco}\\
$f_i(x)=a_ix+b_i$\quad($1> a_i\ge 0,~b_i<0$)&$\O(n\log n)$&\cite{ho1993cos}\\
%$f_i(x)=\min\{ax+b_i,r_i\}$\quad($1>a> 0,~b_i<0$)&$\O(n^6\log n)$&\cite{cheng1998tco}\\
$f_i(x)=\begin{cases}a_ix+b_i&(x\ge d_i)\\-\infty&(x<d_i)\end{cases}$\quad($1>a_i> 0$)&\strongly NP-hard&\cite{cheng1998tco}\\
$f_i(x)=a_ix+b_i$\quad $(a_i\ge 0)$&\(\boldsymbol{\O(n\log n)}\)&\textbf{[Theorem \ref{theorem:partial-total}]}\\
$f_i(x)=\max\{x,a_ix+b_i\}$\quad $(a_i\ge 0)$&\(\boldsymbol{\O(n\log n)}\)&\textbf{[Theorem \ref{theorem:partial-total}]}\\
$f_i(x)=\max\{x,a_ix+b_i,c_i\}$\quad $(a_i\ge 0)$&\(\boldsymbol{\O(n^2)}\)&\textbf{[Theorem \ref{theorem:pconst}]}\\
$f_i(x)=\max\{x,\min\{a_i^1x+b_i^1,a_i^2x+b_i^2\}\}$\quad $(a_i^1,a_i^2>0)$&\textbf{NP-hard}&\textbf{[Theorem \ref{theorem:hardness}]}\\
$f_i(x)=\max\{a_i^1x+b_i^1,a_i^2x+b_i^2\}$\quad $(a_i^1,a_i^2>0)$&\textbf{NP-hard}&\textbf{[Theorem \ref{theorem:hardness}]}\\
\hline
\end{tabular}
\end{center}
\end{table}

\subsection*{The organization of the paper}
The rest of the paper is organized as follows.
In Section \ref{sec:preliminaries}, 
we show that the minimum and/or partial versions of the optimal composition ordering problem can be formulated as the maximum total composition ordering problem.  
In Section~\ref{sec:maxpcop}, we prove the partial composition part of Theorem~\ref{theorem:partial-total} and Theorem \ref{theorem:pconst}, 
and in Section \ref{sec:maxtcop}, we prove the total composition part of Theorem \ref{theorem:partial-total} and Theorem\ref{theorem:exact}.
Finally, Section \ref{sec:nr} provides a proof of Theorem \ref{theorem:hardness}.
%, i.e., 
%the intractable results for  the optimal composition ordering problems, even if $f_i$'s are almost linear. 
%Theorem \ref{theorem:exact}. 
%Due to the space constraint, some of the proofs are omitted, which can be found in Appendix.

\section{Properties of Function Composition}\label{sec:preliminaries}
In this section,
we present two basic properties of the optimal composition ordering problems, which imply that 
the {\em maximum total} composition ordering problem represents all the other composition ordering problems
namely, the  minimum  partial, the minimum total, and the maximum partial ones.   

Let us start with the lemma that the minimization problems are equivalent to the maximization ones.
For a function $f:\mathbb{R} \to \mathbb{R}$, define a function $\tilde{f} :\mathbb{R} \to \mathbb{R}$
by 
\begin{equation}
\label{eq-a1}
\tilde{f}(x):=-f(-x). 
\end{equation}
For example, if $f(x)=2x-3$, then we have $\tilde{f}(x)=2x+3$. 
By the definition, we have $\tilde{\tilde{f}}=f$, and $\tilde{f}$ inherits several properties for $f$, e.g., linearity and monotonicity.

\begin{lemma}\label{lemma:minmax}
Let \(c\) be a real, and for \(i=1,\dots,n\), let \(f_i:\mathbb{R}\to\mathbb{R}\) be real functions.
Then we have the following two statements.
\begin{enumerate}
\item A permutation \(\sigma:[n]\to[n]\) is optimal for 
the maximum total composition ordering problem $((f_i)_{i\in[n]},c)$ if and only if
it is optimal for 
the minimum total composition ordering problem $((\tilde{f}_i)_{i\in[n]},-c)$.

\item A permutation \(\sigma:[n]\to[n]\) and an integer $k$ with  \(0\le k\le n\) form
an optimal solution for 
the maximum partial composition ordering problem $((f_i)_{i\in[n]},c)$
 if and only if they form an optimal solution for 
the minimum partial composition ordering problem $((\tilde{f}_i)_{i\in[n]},-c)$.
\end{enumerate}
\end{lemma}
\begin{proof}
For any permutation \(\sigma:[n]\to[n]\) and an integer $k$ with \(0\le k\le n\), we have
\begin{align*}
f_{\sigma(k)}\circ f_{\sigma(k-1)}\circ\dots\circ f_{\sigma(1)}(c)
&=-\tilde{f}_{\sigma(k)}\circ \tilde{f}_{\sigma(k-1)}\circ\dots\circ \tilde{f}_{\sigma(1)}(-c),
\end{align*}
which proves  the lemma. 
\end{proof}

\noindent
Due to the lemma, this paper deals with the maximum composition ordering problems only. 

We next show the relationships between total and partial compositions.  
For a function $f:\mathbb{R} \to \mathbb{R}$, define a function $\overline{f}:\mathbb{R} \to \mathbb{R}$
by 
\begin{equation}
\label{eq-a2}
\overline{f}(x):=\max\{f_i(x),x\}.
\end{equation}

\begin{lemma}\label{lemma:partialtotal}
Let \(c\) be a real, and for \(i=1,\dots,n\), let \(f_i:\mathbb{R}\to\mathbb{R}\) be real functions.
Then the objective value of the maximum partial composition ordering problem $((f_i)_{i\in[n]},c)$ 
  is equal to the one of the maximum total composition ordering problem $((\overline{f}_i)_{i\in[n]},c)$.  
Moreover, we have the following relationships for the optimal solutions. 
\begin{enumerate}
\item 
If a permutation \(\sigma:[n]\to[n]\) and an integer $k$ with \(0\le k\le n\) form an optimal solution for the maximum partial composition ordering problem $((f_i)_{i\in[n]},c)$, then $\sigma$ is optimal for the maximum total composition ordering problem $((\overline{f}_i)_{i\in[n]},c)$. 
\item Let \(\sigma:[n]\to[n]\) denote an optimal permutation for 
the maximum total composition ordering problem $((\overline{f}_i)_{i\in[n]},c)$. 
Then let $k$ denote the number of $i$'s such that 
\begin{equation}
\label{eq-a3}
\overline{f}_{\sigma(i)}\circ \dots \circ \overline{f}_{\sigma(1)}(c)
>\overline{f}_{\sigma(i-1)}\circ \dots \circ \overline{f}_{\sigma(1)}(c), 
\end{equation}
and $\tau : [n]\to[n]$ denote a permutation such that $\tau(j)$ $(j \leq k)$ is equal to the $j$th $\sigma(i)$ that satisfies {\rm (\ref{eq-a3})}.  
Then $(\tau, k)$
is optimal for the maximum partial composition ordering problem $((f_i)_{i\in[n]},c)$. 
%\begin{align*}
%\{\tau(1),\dots,\tau(k)\}=
%\{\sigma(i)\mid \overline{f}_{\sigma(i)}\circ \dots \circ \overline{f}_{\sigma(1)}(c)
%>\overline{f}_{\sigma(i-1)}\circ \dots \circ \overline{f}_{\sigma(1)}(c)\}
%%%%\label{eq:partialtotal}
%\end{align*}
%where \(\sigma^{-1}(\tau(1))<\dots< \sigma^{-1}(\tau(k))\).
\end{enumerate}
\end{lemma}
\begin{proof}
Let \(\sigma:[n]\to[n]\) be a permutation and \(k\) be a nonnegative integer. 
Then we have 
\begin{align}
\label{eq-qe1}
f_{\sigma(k)}\circ \dots \circ f_{\sigma(1)}(c)
\le \overline{f}_{\sigma(k)}\circ \dots \circ \overline{f}_{\sigma(1)}(c)
\le \overline{f}_{\sigma(n)}\circ \dots \circ \overline{f}_{\sigma(1)}(c)
\end{align}
by \(\overline{f}(x)\ge f(x)\) and \(\overline{f}(x)\ge x\). 
This implies that  the objective value of the maximum partial composition ordering problem $((f_i)_{i\in[n]},c)$ 
 is at most the one of the maximum total composition ordering problem $((\overline{f}_i)_{i\in[n]},c)$.  

On the other hand, for a  permutation \(\sigma:[n]\to[n]\), let  \(\tau\) and \(k\) be defined as the statement in  the lemma. 
Then we have 
\begin{align}
\label{eq-qe2}
f_{\tau(k)}\circ \dots \circ f_{\tau(1)}(c)
= \overline{f}_{\tau(k)}\circ \dots \circ \overline{f}_{\tau(1)}(c)
= \overline{f}_{\sigma(n)}\circ \dots \circ \overline{f}_{\sigma(1)}(c)
\end{align}
by the definition of \(\tau\), which implies that 
the objective value of the maximum partial composition ordering problem $((f_i)_{i\in[n]},c)$ 
  is at least the one of the maximum total composition ordering problem $((\overline{f}_i)_{i\in[n]},c)$.  
Therefore, the objective values of the two problems are same. 

Moreover, this together with (\ref{eq-qe1}) and (\ref{eq-qe2}) implies $(a)$ and $(b)$ in the lemma. 
\end{proof}

From Lemmas \ref{lemma:minmax} and \ref{lemma:partialtotal}, 
it is enough to consider the maximum total composition ordering problem. 
However, the properties of  the functions $f_i$ are not always inherited. 
For example, 
the partial composition ordering problem for the linear functions does not correspond to the total one for the linear functions.

\section{Maximum Partial Composition Ordering Problem}\label{sec:maxpcop}
In this section, we discuss tractable results for  
the maximum partial composition ordering problem for monotone and almost-linear functions. 
By  Lemma~\ref{lemma:partialtotal},  
we deal with the problem as the maximum total composition ordering problem 
for functions $\overline{f}_i$ $(i \in [n])$, 
where \(\overline{f}_i(x)=\max\{f_i(x),x\}\). 
%It is easy to see that
Recall that the objective value of the maximum partial composition ordering problem $((f_i)_{i\in[n]},c)$
is equal to the one of the maximum total composition ordering problem $((\overline{f}_i)_{i\in[n]},c)$.  
Let us start with the maximum partial composition ordering problem for monotone linear functions \(f_i(x)=a_ix+b_i \, (a_i\ge 0)\), i.e., the total composition ordering problem for $\overline{f}_i(x)=\max \{a_ix+b_i,x\}$ $ (a_i\ge 0)$.

The following binary relation $\preceq$ plays an important role for the problem. 
\begin{definition}\label{def:preceq}
For two functions \(f,g:\mathbb{R}\to\mathbb{R}\), 
we write \(f\preceq g\) $($or  \(g\succeq f\)$)$ if \(f\circ g(x)\le g\circ f(x)\)  for any \(x \in \mathbb{R}\), 
 \(f\simeq g\) if  \(f\preceq g\) and  \(f\succeq g\)  $($i.e., \(f\circ g(x)=g\circ f(x)\) for any  \(x \in \mathbb{R}\)$)$, and 
 \(f\prec g\) $($or \(g\succ f\)$)$ if \(f\preceq g\) and \(f\not\simeq g\).
\end{definition}

Note that the relation \(\preceq\) is not total relation in general, here a relation $\preceq$ is called total if $f\preceq g$ or $g\preceq f$ for any $f,g$.
For example, let \(f_1(x)=\max\{2x,3x\}\) and \(f_2(x)=\max\{2x-1,3x+1\}\).
Then $f_1\circ f_2(0) \,(=3)$ is greater than $f_2\circ f_1(0)\,(=1)$, but $f_1\circ f_2(-2)\,(=-10)$ is less than $f_2\circ f_1(-2)\,(=-9)$.

However, if two consecutive functions are total, then we have the following easy but useful lemma.
\begin{lemma}\label{lemma:localchange}
Let \(f_1,\dots,f_n\) be monotone nondecreasing functions.
If \(f_{i}\preceq f_{i+1}\), then it holds that
\(f_n\circ \dots\circ f_{i+2}\circ f_{i+1}\circ f_i\circ f_{i-1}\circ \dots \circ f_1(x)
\ge f_n\circ \dots\circ f_{i+2}\circ f_{i}\circ f_{i+1}\circ f_{i-1}\circ \dots \circ f_1(x)\)
for any \(x \in \mathbb{R}\).
\end{lemma}

It follows from the lemma that, for monotone functions $f_i$, there exists a maximum total composition $f_n \circ f_{n-1} \circ \dots \circ f_1$ that satisfies
 \(f_{1}\preceq f_{2}\preceq \dots \preceq f_{n}\), if the relation is total. 
Moreover, if the relation $\preceq$ is in addition transitive (i.e., $f \preceq g$ and  $g \preceq h$ imply $f \preceq h$),
then it is not difficult to see that 
 \(f_{1}\preceq f_{2}\preceq \dots \preceq f_{n}\) becomes a sufficient condition that $f_n \circ f_{n-1} \circ \dots \circ f_1$ is  a maximum total composition ordering, 
where the proof is given as the more general form in Lemma \ref{lemma:extotalorder}. 

The relation is total if all functions are linear or  of  the form \(\max\{ax+b,x\}\) with $a \geq 0$.
\begin{lemma}
The relation \(\preceq\) is total for linear functions. 
\end{lemma}
\begin{proof}
Let \(f_i(x)=a_ix+b_i\) and \(f_j(x)=a_jx+b_j\). 
Then we have
\begin{align}
f_i\preceq f_j 
&\iff f_i\circ f_j(x)\le f_j\circ f_i(x)  \mbox{ for any }   x \in \mathbb{R} \notag\\
&\iff a_i(a_jx+b_j)+b_i\le a_j(a_ix+b_i)+b_j  \mbox{ for any }   x \in \mathbb{R} \notag\\
%&\iff a_ib_j+b_i\ge a_jb_i+b_j\notag\\
&\iff b_i(1-a_j)\le b_j(1-a_i). \label{eq:lprec}
\end{align}
Since the last inequality consists of the constants only, we have $f_i\preceq f_j$ or $f_i\succeq f_j$.
\end{proof}

\noindent
The totality of the relation is proven in Lemma \ref{lemma:gammaprec-check}, when all functions are  of  the form \(\max\{ax+b,x\}\) with $a \geq 0$.

We further note that the relation \(\preceq\) is transitive for linear functions $f(x)=ax+b$ with $a>1$, 
since (\ref{eq:lprec}) is equivalent to $b_i/(1-a_i)\le b_j/(1-a_j)$, and hence  the ordering $b_1/(1-a_1)\le b_2/(1-a_2) \le \dots \le  b_n/(1-a_n)$ gives an optimal solution for the maximum total composition ordering problem. 
Therefore, it can be solved efficiently by sorting the elements by $b_i/(1-a_i)$. 
The same statement holds when all linear functions have slope less than $1$.
This idea is used for the linear  deterioration  and   linear shortening models for time-dependent scheduling problems. 
However,  in general, this is not the case, i.e.,   the relation \(\preceq\) does not satisfy transitivity. 
Let \(f_1(x)=2x+1\), \(f_2(x)=2x-1\), and \(f_3(x)=x/2\). 
Then we have \(f_1\prec f_2\), \(f_2\prec f_3\), and \(f_3\prec f_1\), 
which implies that the transitivity is not satisfied for linear functions, 
and \(\overline{f}_1\prec \overline{f}_2\), \(\overline{f}_2\prec \overline{f}_3\), and \(\overline{f}_3\prec \overline{f}_1\) hold, implying that 
the transitivity is not satisfied for the functions of the form \(\max\{ax+b,x\}\) with $a \geq 0$. 
% where we recall that $\overline{f}=\max\{f,x\}$. 
These show that the maximum total and partial composition ordering problems are not trivial, even when all functions are monotone and linear.
%In this section, we show that the idea can be generalized to the maximum  partial composition ordering problem for  monotone linear functions.

% we skip the antisymmetry .
%For example, let \(f_1(x)=4x\) and \(f_2(x)=x/2\) (resp. \(\overline{f}_1(x)=\max\{4x,x\}\) and \(\overline{f}_2=\max\{x/2,x\}\)).
%Then we have \(f_1\circ f_2(x)=f_2\circ f_1(x)=2x\) (resp. \(\overline{f}_1\circ \overline{f}_2(x)=\overline{f}_2\circ \overline{f}_1(x)=\max\{x/2,4x\}\)).

%%%%%%%%%%%%%%%%%%%%%%%%%%%%%%%%%%%%%%%%%%%%%%%%%%%%%%%%%%%%%%%%%%%%%%%%%%%%%%%%%%%%%%%%%%%%%%

We first show the following key lemma which can be used even for {\em non-transitive} relations. 
\begin{lemma}\label{lemma:extotalorder}
For monotone nondecreasing  functions $f_i:\mathbb{R} \to \mathbb{R}$ \((i \in [n])\), 
%let \(\preceq\) be the binary relation on \(\{f_1,\dots,f_n\}\)
%defined in Definition \ref{def:preceq}.
if a permutation \(\sigma: [n] \to [n] \) satisfies that $i \le j $  implies $f_{\sigma(i)}\preceq f_{\sigma(j)}$ for any \(i,j\in[n]\),
then \(\sigma\) is an optimal solution for the maximum total composition ordering problem
for \(((f_i)_{i\in[n]},c)\).
%% If there exists a total order\footnote{Total order is a binary relation which is transitive, antisymmetric, and total.} \(\sqsubseteq\) on \(\{f_1,\dots,f_n\}\)
%% such that \(f_i\sqsubseteq f_j \To f_i\preceq f_j\),
%% then the ascending order \(\sigma^*\) accodring to \(\sqsbseteq\)
%% is the optimal solution for the maximum total composition ordering problem
%% for \(f_i\) (\(i\in[n]\)) and \(c\).
\end{lemma}
\begin{proof}
Without loss of generality,
we may assume that \(\sigma\) is the identity permutation.
Let \(\sigma'\) be an optimal solution for the maximum total composition ordering problem such that it has the minimum inversion number.
Here, the inversion number denotes the number of pairs $(i,j)$ with $i<j$ and $\sigma'(i)>\sigma'(j)$. 
Then we show that \(\sigma'\) is the identity permutation by contradiction.
Assume that \(\sigma'(l)>\sigma'(l+1)\) for some $l$. 
Then consider the following permutation:
\begin{align*}
\tau(i)=\begin{cases}
\sigma'(i) &(i\ne l,~l+1),\\
\sigma'(l+1) & (i=l),\\
\sigma'(l)   & (i=l+1). 
\end{cases}
\end{align*}
Since \(\sigma'(l+1)< \sigma'(l)\) implies \(f_{\sigma'(l+1)}\preceq f_{\sigma'(l)}\) by the condition of the identity $\sigma$, 
Lemma \ref{lemma:localchange} implies that $\tau$
is also optimal for the problem. 
Since \(\tau\) has an inversion number smaller than the one for $\sigma'$,  
we derive a contradiction. 
Therefore, \(\sigma'\) is the identity. 
\end{proof}
%Note that the condition of $\sigma$ in the lemma implies the totality of $\preceq$. 
As mentioned above, if the relation \(\preceq\) is in addition transitive (i.e., \(\preceq\) is a total preorder), then 
such a $\sigma$ always exists. 

To efficiently solve the maximum partial composition ordering problem for the linear functions, we show that 
for $\overline{f}_i(x)=\max \{a_ix+b_i,x\}$ $ (a_i\ge 0)$, 
(i) there exists a permutation \(\sigma\) which satisfies the condition in Lemma \ref{lemma:extotalorder} and
(ii) the permutation $\sigma$ can be computed efficiently. 
We analyze the relation $\preceq$ in terms of the following $\gamma$ and $\delta$,
and provide an efficient algorithm. 
\begin{definition}
For a linear function \(f(x)=ax+b\), we define 
\begin{align*}
\gamma(f)=
\begin{cases}
\frac{b}{1-a}&(a\ne 1),\\
+\infty&(a=1~\text{and}~b<0),\\
-\infty&(a=1~\text{and}~b\ge 0)
\end{cases}
\quad\text{and}\quad
\delta(f)=
\begin{cases}
+1&(a\ge 1),\\
-1&(a<1).
\end{cases}
\end{align*}
\end{definition}
\noindent Note that \(\gamma(f)\) is the solution of the equation \(f(x)=x\) if  $\gamma(f)\not= -\infty, +\infty$. 
In the rest of the paper, we assume without loss of generality that no $f_i$ is identity (i.e., \(f_i(x)=x\)), since we can ignore identity function for both the total and partial composition problems.

Let $\sigma:[n]\to [n]$ denote a permutation that is {\em compatible} with the lexicographic ordering with respect to $(\delta(f_i), \gamma(f_i))$, i.e., 
 $(\delta(f_{\sigma(i)}), \gamma(f_{\sigma(i)}))$ is lexicographically smaller than or equal to $(\delta(f_{\sigma(j)}), \gamma(f_{\sigma(j)}))$ if $i<j$.
Namely, there exists an integer $k$ such that $0\leq k \leq n$,  
\(\delta(f_{\sigma(1)})=\dots=\delta(f_{\sigma(k)})=-1\),
\(\delta(f_{\sigma(k+1)})=\dots=\delta(f_{\sigma(n)})=+1\),
 \(\gamma(f_{\sigma(1)})\le\dots\le\gamma(f_{\sigma(k)})\),  and 
\(\gamma(f_{\sigma(k+1)})\le\dots\le\gamma(f_{\sigma(n)})\).

We prove that the lexicographic order satisfies the condition in Lemma \ref{lemma:extotalorder} and
thus, the order is the optimal solution.
\begin{lemma}
\label{lemma-ee}
For monotone nondecreasing linear functions $f_i$ $(i \in [n])$, 
let $\sigma$ denote a permutation compatible with the lexicographic order with respect to $(\delta(f_i), \gamma(f_i))$. 
Then $i \le j$ implies  $\overline{f}_{\sigma(i)}\preceq \overline{f}_{\sigma(j)}$ for any \(i,j\in[n]\).
\end{lemma}
Before proving Lemma \ref{lemma-ee}, we discuss algorithms for the maximum partial composition ordering problem.

\subsection{Algorithms}
By Lemma \ref{lemma-ee},
%the maximum partial composition ordering problem for the monotone nondecreasing linear functions $f_i$, equivalently, 
the maximum total composition ordering problem $((\overline{f}_i)_{i\in[n],c})$ such that $f_i$'s are monotone nondecreasing linear functions can be solved by computing the lexicographic order with respect to $(\delta(f_i), \gamma(f_i))$. 
Therefore, it can be solved in \(\O(n\log n)\) time.
This is our algorithm for the partial composition part of Theorem \ref{theorem:partial-total}.
We remark that the time complexity \(\O(n\log n)\) of the problem is the best possible  in the comparison model. 
We also remark that the optimal value for the maximum partial composition ordering problem for \(f_i(x)=a_ix+b_i\) \((a_i\ge 0)\) forms a piecewise linear function (in $c$) with at most \((n+1)\) pieces.

\medskip
Next, for $i \in [n]$,
let \(f_i(x)=a_ix+b_i\) be a monotone nondecreasing linear function
and let $h_i(x)=\max\{f_i(x),c_i\}$ for some constant $c_i$. 
We give an efficient algorithm for
the maximum partial composition ordering problem $((h_i)_{i\in[n]},c)$,
which is the tractability result of Theorem \ref{theorem:pconst}. 
As mentioned in the introduction, 
the problem includes 
the two-valued free-order secretary problem, 
and it is a generalization of the maximum partial composition ordering problem for monotone linear functions. 

By Lemma \ref{lemma:partialtotal}, 
we instead consider the maximum total composition ordering problem for the functions 
\begin{equation}
\label{eq-111a}
\overline{h}_i(x) =\max\{a_ix+b_i, c_i, x\} \mbox{ for } a_i\in\mathbb{R}_+,~ b_i,c_i \in\mathbb{R},
\end{equation}
where \(\mathbb{R}_+\) is the set of nonnegative real numbers.

\begin{lemma}
\label{lemma-1ab}
Let $c\in\mathbb{R}$, and let $\overline{h}_i$ $(i \in [n])$ be a function defined as $(\ref{eq-111a})$.
Then there exists an optimal solution \(\sigma\)
for the maximum total composition ordering problem $((\overline{h}_i)_{i\in[n]},c)$ 
such that no $i\,(>1)$ satisfies \(\overline{f}_{\sigma(i)}\circ\overline{h}_{\sigma(i-1)}\circ\dots\circ \overline{h}_{\sigma(1)}(c)< c_{\sigma(i)}\),
where $\overline{f}_i(x)=\max\{a_ix+b_i,x\}$. 
\end{lemma}
\begin{proof}
Let $\sigma$ denote an optimal solution for the problem. 
Assume that there exists an index \(i\) that satisfies the condition in the lemma. 
Let $i^*$ denote the largest such $i$. 
Then by the definition of $i^*$, we have $\overline{h}_{\sigma(i^*)}\circ\dots\circ \overline{h}_{\sigma(1)}(c)=\overline{h}_{\sigma(i^*)}(c)=c_{\sigma(i^*)}$. 
It holds that 
 \(c_{\sigma(i)}< c_{\sigma(i^*)}\) for any $i$ with \(0\leq i< i^*\), 
since 
\(c_{\sigma(i)}\le \overline{h}_{\sigma(i)}\circ\dots\circ \overline{h}_{\sigma(1)}(c)\le \overline{h}_{\sigma(i^*-1)}\circ\dots\circ \overline{h}_{\sigma(1)}(c)< c_{\sigma(i^*)}\), 
where $c_{\sigma(0)}=c$ is assumed. 
Thus, we have
\(
\overline{h}_{\sigma(n)}\circ\dots\circ \overline{h}_{\sigma(1)}(c)
= \overline{h}_{\sigma(n)}\circ\dots\circ \overline{h}_{\sigma(i^*)}(c)
\le \overline{h}_{\sigma(i^*-1)}\circ\dots\circ \overline{h}_{\sigma(1)}\circ \overline{h}_{\sigma(n)}\circ\dots\circ \overline{h}_{\sigma(i^*)}(c).
\)
This implies that $(\sigma(i^*), \dots \sigma(n),\sigma(1). \dots , \sigma(i^*-1))$ is also an optimal permutation for the problem. 
Moreover, in the composition according to this permutation,
the constant part of $\bar{h}_i$ $(i\not=i^*)$ is not explicitly used by the definition of $i^*$
 and \(c_{\sigma(i)}< c_{\sigma(i^*)}\) for any $i \,(< i^*)$, which completes the proof. 
\end{proof}

%Proof of Theorem \ref{theorem:pconst}
It follows from Lemma \ref{lemma-1ab} that 
 an optimal solution for the problem can be obtained by solving the following $n+1$ instances of the maximum partial composition ordering problem for monotone nondecreasing linear functions $((f_i)_{i\in[n]},c)$ and $((f_i)_{i\in[n]\setminus\{k\}},c_k)$  for all $k\in[n]$. 

Therefore, we have an \(\O(n^2 \log n)\)-time algorithm by directly applying Theorem \ref{theorem:partial-total} to the problems.
Moreover, we note that
the maximum partial composition ordering problem for monotone nondecreasing linear functions
can be solved in linear time if we know the lexicographic order.
This implies that the problem can be solved in  \(\O(n^2)\) time by first computing the lexicographic order with respect to  $(\delta(f_i), \gamma(f_i))$.  
This is our algorithm for Theorem \ref{theorem:pconst}.

\subsection{The proof of Lemma \ref{lemma-ee}}
In this subsection, we prove Lemma \ref{lemma-ee}.
We first consider the relationship between two linear functions.
The proof can be found in Appendix.
\begin{lemma}\label{lemma:gammaprec}
Let \(f_i(x)=a_ix+b_i\) and \(f_j(x)=a_jx+b_j\) be $($non-identity$)$ monotone nondecreasing functions $($i.e., \((a_i,b_i),(a_j,b_j)\ne (1,0)\), \(a_i,a_j\ge 0\)$)$. 
Then we have the following statements;
\begin{enumerate}
\item if \(a_i,a_j=1\), then \(f_i\simeq f_j\),  \label{lemma:gammaprec11}
\item if \(a_i,a_j\ge 1\) and \(a_i\cdot a_j>1\), then \(f_i\preceq f_j\siff\gamma(f_i)\le \gamma(f_j)\),  \label{lemma:gammaprec++}
\item if \(a_i,a_j<1\), then \(f_i\preceq f_j\siff\gamma(f_i)\le \gamma(f_j)\), \label{lemma:gammaprec--}
\item if \(a_i\ge 1\), \(a_j<1\), then \(f_i\preceq f_j\siff\gamma(f_i)\ge \gamma(f_j)\) and \(f_i\succeq f_j\siff\gamma(f_i)\le \gamma(f_j)\).\label{lemma:gammaprec+-}
%\item if \(a_i< 1\), \(a_j\ge 1\), then \(f_i\preceq f_j\siff\gamma(f_i)\ge \gamma(f_j)\) and \(f_i\simeq f_j\siff\gamma(f_i)=\gamma(f_j)\).\label{lemma:gammaprec-+}
\end{enumerate}
\end{lemma}
By this lemma, the relation \(\preceq\) is total preorder for the both cases \(a_1,a_2,\dots,a_n\ge 1\) and \(a_1,a_2,\dots,a_n<1\).
Moreover, the permutation \(\sigma:[n]\to [n]\) such that \(\gamma(f_{\sigma(1)})\le\dots\le\gamma(f_{\sigma(n)})\) is optimal for the cases.
This result matches the results 
in the time-dependent scheduling problem
of the linear deterioration model (when \(a_1,a_2,\dots,a_n\ge 1\))
and the linear shortening model (when \(a_1,a_2,\dots,a_n< 1\)).

Next we characterize the relationship between two functions of the form \(\max\{a_ix+b_i,x\}\),
the proof can be found in Appendix.
\begin{lemma}\label{lemma:gammaprec-check}
For $($non-identity$)$ monotone nondecreasing linear functions \(f_i(x)=a_ix+b_i\) and \(f_j(x)=a_jx+b_j\),  
we have the following statements;
\begin{enumerate}
\item if \(a_i, a_j \geq 1\) and \(\gamma(f_i)\le \gamma(f_j)\), then \(\overline{f}_i\preceq \overline{f}_j\),\label{lemma:gammaprec-check++}
\item if \(a_i, a_j < 1\)  and \(\gamma(f_i)\le \gamma(f_j)\), then \(\overline{f}_i\preceq \overline{f}_j\), \label{lemma:gammaprec-check--}
\item if \(a_i <1 \), \(a_j \geq 1\), and \(\gamma(f_i)\le \gamma(f_j)\), then \(\overline{f}_i\simeq \overline{f}_j\),\label{lemma:gammaprec-check-+}
\item if \(a_i \geq 1\), \(a_j<1\), and \(\gamma(f_i)\le \gamma(f_j)\), then \(\overline{f}_i\succeq \overline{f}_j\).\label{lemma:gammaprec-check+-}
\end{enumerate}
\end{lemma}

Note that Lemma \ref{lemma:gammaprec-check} implies that the relation $\preceq$ is total for the functions of the form \(\max \{ax+b,x\}\) with $a\geq 0$. 
Moreover, it is easy to check that the lexicographic order with respect to  $(\delta(f_i), \gamma(f_i))$
satisfies the condition in Lemma \ref{lemma:extotalorder}, i.e.,
$i<j$ implies \(\overline{f}_{\sigma(i)}\preceq \overline{f}_{\sigma(j)}\)
for the permutation $\sigma$ that is compatible with the ordering.
Therefore, we have Lemma \ref{lemma-ee}.

\section{Maximum Total Composition and Exact $k$-composition Ordering Problems}\label{sec:maxtcop}
In this section we prove the total composition part of Theorem \ref{theorem:partial-total} and Theorem \ref{theorem:exact}.
Different from the case when each function is of form $\max\{a_ix+b_i,x\}$,
the binary relation $\preceq$ for linear functions does not satisfy the condition in Lemma \ref{lemma:extotalorder}.
In fact, we do not know if there exists such a simple criterion on the maximum total composition ordering.  
We instead present an efficient algorithm that chooses the best ordering among linearly many candidates.
Our main result is the following lemma.
\begin{lemma}
\label{lemma:tc}
For monotone nondecreasing linear functions $f_i$ $(i \in [n])$, 
let $\sigma$ denote a permutation compatible with the lexicographic order with respect to $(\delta(f_i), \gamma(f_i))$. 
Then an optimal solution for the maximum total composition ordering problem \(((f_i)_{i\in[n]},c)\) is
\begin{align*}
  (\sigma(t),\sigma(t+1),\dots,\sigma(n),\sigma(1),\sigma(2),\dots,\sigma(t-1))
\end{align*}
for some $t$.
\end{lemma}

Before proving Lemma \ref{lemma:tc}, we discuss algorithms for the maximum total composition and the exact $k$-composition ordering problems.

\subsection{Algorithm for Total Composition}
In this subsection, we prove the total composition part of Theorem \ref{theorem:partial-total}, i.e.,
we provide an efficient algorithm for the maximum total composition ordering problem \(((f_i(x)=a_ix+b_i)_{i\in[n]},c)\), where \(a_i\ge 0\).
Let \(\sigma:[n]\to[n]\) be a permutation compatible with the lexicographic order with respect to $(\delta(f_i), \gamma(f_i))$.
Then there exists an optimal solution of the form 
\((\sigma(t),\sigma(t+1),\dots,\sigma(n),\sigma(1),\sigma(2),\dots,\sigma(t-1))\) for some $t$ by Lemma \ref{lemma:tc}.
Therefore, the problem can be computed in polynomial time by checking $n$ permutations above. 
To reduce the time complexity, let $d_k=f_{\sigma(k-1)}\circ\dots\circ f_{\sigma(1)}\circ f_{\sigma(n)}\circ\dots\circ f_{\sigma(k)}(c)$ for $k=1,\dots, n$ and let \(a=\prod_{i=1}^n a_i\).
Then it is not difficult to see that 
\(d_{k+1}=a_{\sigma(k)}\cdot (d_k-a\cdot c)-b_{\sigma(k)}\cdot(a-1)+a \cdot c, \)
and hence the problem is solvable in $\O(n\log n)$ time.

%% Algorithm \ref{alg:tc} solves the problem in $\O(n \log n)$ time. 
%% The correctness of the algorithm is shown in the next subsection.
%% \begin{algorithm}
%% \caption{Maximum Total Composition}\label{alg:tc}
%% \begin{algorithmic}[1]
%% \STATE compute a permutation $\sigma$ which is compatible with the lexicographic order with respect to $(\delta(f_i), \gamma(f_i))$.
%% \STATE \(d_1\ot f_{\sigma(n)}\circ\cdots\circ f_{\sigma(1)}(c)\)
%% \STATE \(a\ot a_1\cdot a_2\cdot\dots\cdot a_n\)
%% \STATE \(s\ot d_1\), \(t\ot 1\)
%% \FOR{$i=1$ to $n-1$}
%% \STATE \(d_{i+1}\ot a_{\sigma(i)}\cdot (d_i-a\cdot c)-b_{\sigma(i)}\cdot (a-1)+a\cdot c\)
%% \ENDFOR
%% \STATE return \((\sigma(t),\sigma(t+1),\dots,\sigma(n),\sigma(1),\sigma(2),\dots,\sigma(t-1))\) with \(t\in\argmax\{d_i\mid i\in[n]\}\)
%% \end{algorithmic}
%% \end{algorithm}

\subsection{exact $k$-composition}
In this subsection, we prove Theorem \ref{theorem:exact}, i.e.,
we provide an efficient algorithm for the maximum exact $k$-composition ordering problem \(((f_i(x)=a_ix+b_i)_{i\in[n]},c)\), where \(a_i\ge 0\).
We use a dynamic programming to find the optimal value.

For simplicity, we relabel the indices of functions
so that the lexicographic order of $(\delta(f_i), \gamma(f_i))$ is monotone increasing.
%% i.e., \(\delta(f_{1})=\dots=\delta(f_{r})=-1\),
%% \(\delta(f_{r+1})=\dots=\delta(f_{n})=1\),
%% \(\gamma(f_{1})\le\dots\le\gamma(f_{r})\),
%% and \(\gamma(f_{r+1})\le\dots\le\gamma(f_{n})\).
We use dynamic programming to solve the problem.
Let $m(i,j,l)$ be the maximum value of $f_{\sigma(l)}\circ\dots\circ f_{\sigma(1)}(c)$
for a permutation $\sigma$
such that $i\le \sigma(1)<\sigma(2)<\dots<\sigma(l)\le i+j-1$ if $i+j-1\le n$,
and $i\le \sigma(1)<\dots<\sigma(p)\le n,~1\le \sigma(p+1)<\dots<\sigma(l)\le i+j-1-n$ for some $p~(0\le p\le l)$ if $i+j-1>n$.
We claim that the optimal value for the problem is $\max_{i=1}^n m(i,n,k)$.

Let $\sigma^*: [n]\to [n]$ be an optimal permutation for the problem.
By Lemma \ref{lemma:tc},
we can assume that
$i^*\le \sigma^*(1)<\dots<\sigma^*(p)\le n,~1\le \sigma^*(p+1)<\dots<\sigma^*(k)\le i^*-1$
for some $i^*$ and $p$.
Therefore, we have $f_{\sigma^*(k)}\circ\dots\circ f_{\sigma^*(1)}(c)\le m(i^*,n,k)\le \max_{i=1}^n m(i,n,k)$ and, thus, $\max_{i=1}^n m(i,n,k)$ is the optimal value for the problem.

For each $i,j,l$, the value $m(i,j,l)$ satisfies the following relation:
\begin{align*}
  m(i,j,l)=
  \begin{cases}
    c & (l=0),\\
    f_{j}(m(i,j-1,l-1)) &(l\ge 1, j=l),\\
    \max\{m(i,j-1,l),f_{j}(m(i,j-1,l-1))\} &(l\ge 1, j>l).\\
  \end{cases}
\end{align*}
To evaluate $\max_{i=1}^n m(i,n,k)$,
our algorithm calculate the values of $m(i,j,l)$ for $0\le i,j\le n$ and $0\le l\le k$.
Therefore, we can obtain the optimal value for the problem in $\O(k\cdot n^2)$ time.
The detailed algorithm for the maximum exact $k$-composition problem is shown in Algorithm \ref{alg:exact}. % where the indices are taken modulo $n$.

\begin{algorithm}
\caption{Maximum Exact $k$-Composition}\label{alg:exact}
\begin{algorithmic}[1]
\STATE sort the input functions according to the lexicographic order of $(\delta(f_i), \gamma(f_i))$
\FOR{\(l=0\) to \(k\)}
\FOR{\(j=l\) to \(n\)}
\FOR{\(i=1\) to \(n\)}
\STATE \textbf{if} $l=0$ \textbf{then} $m(i,j,l)\ot c$
\STATE \textbf{else if} $j=l$ \textbf{then} $m(i,j,l)\ot f_{(i+j \mod{n})}(m(i,j-1,l-1))$
\STATE \textbf{else} $m(i,j,l)\ot \max\{m(i,j-1,l),f_{(i+j\mod{n})}(m(i,j-1,l-1))\}$ 
\ENDFOR
\ENDFOR
\ENDFOR
\STATE return $\max_{i=1}^n m(i,n,k)$
\end{algorithmic}
\end{algorithm}

\subsection{The proof of Lemma \ref{lemma:tc}}
In this subsection, we prove Lemma \ref{lemma:tc}.
To overcome the difficulty that 
the binary relation $\preceq$ for linear functions does not satisfy the condition in Lemma \ref{lemma:extotalorder}, 
we discuss relationships among three or four functions.

The following lemma shows the relationships between
\(\gamma(f_i),\,\gamma(f_j)\), \(\gamma(f_j\circ f_i)\) and \(\gamma(f_i\circ f_j)\) 
for monotone linear functions.
The proof can be found in Appendix.
\begin{lemma}\label{lemma:gammacomp}
For monotone nondecreasing linear functions \(f_i(x)=a_ix+b_i\) and \(f_j(x)=a_j x+b_j\) \((a_i,a_j\ge 0)\),
we have the following statements. 
\begin{enumerate}
\item If \(\gamma(f_i)=\gamma(f_j)\), then \(\gamma(f_i)=\gamma(f_j)= \gamma(f_j\circ f_i)\),\label{lemma:gammacompeq}
\item If \(\gamma(f_i)<\gamma(f_j)\) and \(a_i,a_j\ge 1\), then  \(\gamma(f_i)\le \gamma(f_j\circ f_i)\le\gamma(f_j)\),\label{lemma:gammacomp++}
\item If \(\gamma(f_i)<\gamma(f_j)\) and \(a_i,a_j<1\), then  \(\gamma(f_i)\le \gamma(f_j\circ f_i)\le\gamma(f_j)\),\label{lemma:gammacomp--}

\item If \(\gamma(f_i)<\gamma(f_j)\), \(a_i<1\), \(a_j\ge 1\), and \(a_i\cdot a_j\ge 1\),
then \(\gamma(f_j\circ f_i)\ge\gamma(f_j)~(>\gamma(f_i))\),\label{lemma:gammacomp-+>}

\item If \(\gamma(f_i)<\gamma(f_j)\), \(a_i<1\), \(a_j\ge 1\), and \(a_i\cdot a_j< 1\), 
then  \(\gamma(f_j\circ f_i)\le\gamma(f_i)~(<\gamma(f_j))\),\label{lemma:gammacomp-+<}

\item If \(\gamma(f_i)<\gamma(f_j)\), \(a_i\ge 1\), \(a_j<1\), and \(a_i\cdot a_j\ge 1\),
then  \(\gamma(f_j\circ f_i)\le\gamma(f_i)~(<\gamma(f_j))\),\label{lemma:gammacomp+->}
\item If \(\gamma(f_i)<\gamma(f_j)\), \(a_i\ge 1\), \(a_j<1\), and \(a_i\cdot a_j< 1\),
then \(\gamma(f_j\circ f_i)\ge\gamma(f_j)~(>\gamma(f_i))\).\label{lemma:gammacomp+-<}
\end{enumerate}
% where we allow to write \(+\infty\le +\infty\), \(-\infty\le +\infty\), and 
% \(-\infty\le -\infty\).
\end{lemma}

By Lemmas \ref{lemma:gammaprec} and \ref{lemma:gammacomp},
we have the following inequalities for compositions of four functions.
\begin{lemma}\label{lemma:comp+-+-}
For monotone nondecreasing linear functions \(f_i(x)=a_ix+b_i\)  $(i=1,2,3,4)$, 
if $a_1,a_3 \geq 1$, $a_2,a_4 < 1$ and  \(\gamma(f_1) \ge \gamma(f_2) \ge \gamma(f_3) \ge \gamma(f_4)\), 
then we have
\begin{align*}
f_4\circ f_3\circ f_2\circ f_1(x) \le \max\{f_4\circ f_1\circ f_3\circ f_2(x),~f_3\circ f_2\circ f_4\circ f_1(x)\}\quad(\forall x).
\end{align*}
\end{lemma}
\begin{lemma}\label{lemma:comp-+-+}
For monotone nondecreasing linear functions \(f_i(x)=a_ix+b_i\)  $(i=1,2,3, 4)$, 
if $a_1,a_3<1$, $a_2,a_4 \geq 1$ and  \(\gamma(f_1) \ge \gamma(f_2) \ge \gamma(f_3) \ge \gamma(f_4)\), 
then we have
\begin{align*}
f_4\circ f_3\circ f_2\circ f_1(x)\le \max\{f_4\circ f_1\circ f_3\circ f_2(x),~f_3\circ f_2\circ f_4\circ f_1(x)\}\quad(\forall x).
\end{align*}
%The equality holds only when \(\gamma_i=\gamma_j=\gamma_k=\gamma_l\).
\end{lemma}
\begin{proof}
We only prove Lemma \ref{lemma:comp+-+-} since Lemma \ref{lemma:comp-+-+} can be proved in a similar way.
Let \(g(x)=f_3\circ f_2(x)\).
If \(a_2 \cdot a_3\ge 1\), then \(\gamma(g)\le \gamma(f_3)\le \gamma(f_1)\) holds
by \ref{lemma:gammacompeq} and \ref{lemma:gammacomp+->} in Lemma \ref{lemma:gammacomp},
and \(g\circ f_1(x)\le f_1\circ g(x)\) holds by \ref{lemma:gammaprec11} and \ref{lemma:gammaprec++} in Lemma \ref{lemma:gammaprec}.
Thus, we have \(f_4\circ f_3\circ f_2\circ f_1(x)\le f_4\circ f_1\circ f_3\circ f_2(x)\).

On the other hand, if \(a_2\cdot a_3< 1\), then \(\gamma(g)\ge \gamma(f_2)\ge \gamma(f_4)\) holds by \ref{lemma:gammacompeq} and \ref{lemma:gammacomp+-<} in Lemma \ref{lemma:gammacomp},
and \(f_4\circ g(x)\le g\circ f_4(x)\) holds by \ref{lemma:gammaprec--} in Lemma \ref{lemma:gammaprec}.
Thus, we have \(f_4\circ f_3\circ f_2\circ f_1(x)\le f_3\circ f_2\circ f_4\circ f_1(x)\).
\end{proof}

By Lemmas \ref{lemma:comp+-+-} and \ref{lemma:comp-+-+}, we obtain the following lemma. %(the formal proof can be found in Appendix).
\begin{lemma}\label{lemma:atmost2}
There exists an optimal permutation \(\sigma\) 
for the maximum total composition ordering problem for monotone nondecreasing functions $f_i$ $(i\in [n])$ 
such that
at most two $i$'s 
satisfy \(\delta(f_{\sigma(i)}) \cdot \delta(f_{\sigma(i+1)})=-1\).
\end{lemma}
\begin{proof}
Let \(\sigma\) be an  optimal solution,
with the minimum number of $i$'s satisfying 
$\delta(f_{\sigma(i)})\cdot \delta(f_{\sigma(i+1)})=-1$. 
Assume that $\sigma$ contains at least three such $i$'s. 
Let $i_1,i_2$ and $i_3$ denote the three smallest such $i$'s with $i_1< i_2< i_3$,  
and \(i_4\) denote the fourth smallest such $i$ if exists; otherwise we define $i_4=n$.  
Let \(g_1(x)=f_{\sigma(i_1)}\circ\dots\circ f_{\sigma(1)}(x)\),
\(g_2(x)=f_{\sigma(i_2)}\circ\dots\circ f_{\sigma(i_1+1)}(x)\),
\(g_3(x)=f_{\sigma(i_3)}\circ\dots\circ f_{\sigma(i_2+1)}(x)\),
and \(g_4(x)=f_{\sigma(i_4)}\circ\dots\circ f_{\sigma(i_3+1)}(x)\).
Then we have \(\delta(g_1)= -\delta(g_2)= \delta(g_3)= -\delta(g_4)\).
We claim that \(\gamma(g_1)\ge \gamma(g_2)\ge \gamma(g_3)\ge \gamma(g_4)\).

Assume that \(\gamma(g_j)<\gamma(g_{j+1})\) for some \(j\in\{1,2,3\}\).
Then it follows from \ref{lemma:gammaprec+-} in Lemma
\ref{lemma:gammaprec}  that 
\(g_{j+1}\circ g_{j}(x) \le g_{j}\circ g_{j+1}(x)\) holds,  
 which contradicts the assumption on $\sigma$.
Therefore we have
\begin{align*}
f_{\sigma(n)}\circ\dots\circ f_{\sigma(1)}(x)
&=f_{\sigma(n)}\circ\dots\circ f_{\sigma(i_4+1)}\circ g_4\circ g_3\circ g_2\circ g_1(x)\\
&\le \max\left\{
\begin{matrix}
f_{\sigma(n)}\circ\dots\circ f_{\sigma(i_4+1)}\circ g_4\circ g_1\circ g_3\circ g_2(x),\\
f_{\sigma(n)}\circ\dots\circ f_{\sigma(i_4+1)}\circ g_3\circ g_2\circ g_4\circ g_1(x)
\end{matrix}
\right\}
\end{align*}
by Lemmas \ref{lemma:comp+-+-} and \ref{lemma:comp-+-+}.
This again contradicts the assumption on $\sigma$. 
\end{proof}

Next, we provide inequalities for compositions of three functions.
\begin{lemma}\label{lemma:comp+-+}
For monotone nondecreasing linear functions \(f_i(x)=a_ix+b_i\)  $(i=1,2,3)$, 
if $a_1,a_3\geq 1$, $a_2<1$, \(a_1\cdot a_2\cdot a_3\ge 1\) and \(\gamma(f_1) \ge \gamma(f_2) \ge \gamma(f_3)\), 
then we have
\begin{align*}
f_3\circ f_2\circ f_1(x)\le \max\{f_2\circ f_1\circ f_3(x),~f_1\circ f_3\circ f_2(x)\}\quad(\forall x).
\end{align*}
%The equality holds only when \(\gamma(f_i)=\gamma(f_j)=\gamma(f_k)\).
\end{lemma}
\begin{lemma}\label{lemma:comp-+-}
For monotone nondecreasing linear functions \(f_i(x)=a_ix+b_i\)  $(i=1,2,3)$, 
if $a_1,a_3< 1$, $a_2\geq 1$, \(a_1\cdot a_2\cdot a_3< 1\) and \(\gamma(f_1) \ge \gamma(f_2) \ge \gamma(f_3)\), 
 then we have
\begin{align*}
f_3\circ f_2\circ f_1(x)\le \max\{f_2\circ f_1\circ f_3(x),~f_1\circ f_3\circ f_2(x)\}\quad(\forall x).
\end{align*}
\end{lemma}
\begin{proof}
We only prove Lemma \ref{lemma:comp+-+} since Lemma \ref{lemma:comp-+-} can be prove in a similar way.
If \(a_2\cdot a_3\ge 1\), then \(\gamma(f_3\circ f_2)\le \gamma(f_3)\le \gamma(f_1)\)
by  \ref{lemma:gammacompeq} and \ref{lemma:gammacomp+->} in Lemma \ref{lemma:gammacomp},
and it implies \(f_3\circ f_2\circ f_1(x)\le f_1\circ f_3\circ f_2(x)\)
by \ref{lemma:gammaprec11} and \ref{lemma:gammaprec++} in Lemma \ref{lemma:gammaprec}.
If \(a_2\cdot a_3< 1\) and \(\gamma(f_3\circ f_2)\ge \gamma(f_1)\),
then \(f_3\circ f_2\circ f_1(x)\le f_1\circ f_3\circ f_2(x)\)
by \ref{lemma:gammaprec+-} in Lemma \ref{lemma:gammaprec}.

If \(a_1\cdot a_2\ge 1\), then \(\gamma(f_2\circ f_1)\ge \gamma(f_1)\ge \gamma(f_3)\)
by  \ref{lemma:gammacompeq} and \ref{lemma:gammacomp-+>} in Lemma \ref{lemma:gammacomp},
and it implies \(f_3\circ f_2\circ f_1(x)\le f_2\circ f_1\circ f_3(x)\) by \ref{lemma:gammaprec11} and \ref{lemma:gammaprec++} in Lemma \ref{lemma:gammaprec}.
If \(a_1\cdot a_2< 1\) and \(\gamma(f_2\circ f_1)\le \gamma(f_3)\),
then \(f_3\circ f_2\circ f_1(x)\le f_2\circ f_1\circ f_3(x)\) by \ref{lemma:gammaprec+-} in Lemma \ref{lemma:gammaprec}.

Otherwise, we have
\(a_2\cdot a_3< 1\), \(a_1\cdot a_2<1\),
\(\gamma(f_3\circ f_2)<\gamma(f_1)\),
and \(\gamma(f_2\circ f_1)>\gamma(f_3)\).
Then we have 
\(\gamma((f_3\circ f_2)\circ f_1)\ge \gamma(f_1)\) by \ref{lemma:gammacomp-+>} in Lemma \ref{lemma:gammacomp},
and \(\gamma(f_3\circ (f_2\circ f_1))\le \gamma(f_3)\) by \ref{lemma:gammacomp+->} in Lemma \ref{lemma:gammacomp}
since \(a_1\cdot a_2\cdot a_3\ge 1\).
Therefore \(\gamma(f_1)=\gamma(f_2)=\gamma(f_3)\),
This together with \(\gamma(f_3\circ f_2)< \gamma(f_1)\) contradicts 
\ref{lemma:gammacompeq}  in Lemma \ref{lemma:gammacomp}. 
\end{proof}

By Lemmas \ref{lemma:gammaprec}, \ref{lemma:gammacomp}, \ref{lemma:atmost2}, \ref{lemma:comp+-+}, and \ref{lemma:comp-+-},
we get the following lemmas.
\begin{lemma}\label{lemma:prod>}
If \(\prod_{i=1}^n a_i\ge 1\), then
there exists an optimal permutation \(\sigma\) such that, 
for some two integers \(s,t\) \((0\le s\le t\le n)\), 
\(\delta(f_{\sigma(t+1)})=\dots=\delta(f_{\sigma(n)})=\delta(f_{\sigma(1)})=\dots=\delta(f_{\sigma(s)})=-1\), 
 \(\delta(f_{\sigma(s+1)})=\dots=\delta(f_{\sigma(t)})=1\), 
\(\gamma_{\sigma(t+1)}\le\dots\le\gamma_{\sigma(n)}\leq
\gamma_{\sigma(1)}\le\dots\le \gamma_{\sigma(s)}\),  
and 
\(\gamma_{\sigma(s+1)}\le\dots\le\gamma_{\sigma(t)}\).  
\end{lemma}
\begin{lemma}\label{lemma:prod<}
If \(\prod_{i=1}^n a_i< 1\),
then there exists an optimal permutation \(\sigma\) such that, for some two integers \(s,t\)
\((0\le s\le t\le n)\), 
\(\delta(f_{\sigma(t+1)})=\dots=\delta(f_{\sigma(n)})=\delta(f_{\sigma(1)})=\dots=\delta(f_{\sigma(s)})=1\), 
\(\delta(f_{\sigma(s+1)})=\dots=\delta(f_{\sigma(t)})=-1\), 
 \(\gamma_{\sigma(t+1)}\le\dots\le\gamma_{\sigma(n)} \leq \gamma_{\sigma(1)}\le\dots\le \gamma_{\sigma(s)}\), and
\(\gamma_{\sigma(s+1)}\le\dots\le\gamma_{\sigma(t)}\).
\end{lemma}

\begin{proof}
We only prove Lemma \ref{lemma:prod>} since Lemma \ref{lemma:prod<} can be proved in a similar way.
By Lemma \ref{lemma:atmost2},
there exists an optimal permutation \(\sigma\) and two integers \(s,t\) \((0\le s\le t\le n)\)
such that
\(\delta(f_{\sigma(1)})=\dots=\delta(f_{\sigma(s)})=-\delta(f_{\sigma(s+1)})=\dots=-\delta(f_{\sigma(t)})=\delta(f_{\sigma(t+1)})=\dots=\delta(f_{\sigma(n)})\).
By Lemma \ref{lemma:gammaprec}, we have 
\begin{align*}
\gamma_{\sigma(1)}\le\dots\le \gamma_{\sigma(s)},~
\gamma_{\sigma(s+1)}\le\dots\le \gamma_{\sigma(t)},~
\gamma_{\sigma(t+1)}\le\dots\le \gamma_{\sigma(n)}.
\end{align*}
This implies that the lemma holds when \(s=0\) or \(t=n\). 
For \(0<s\le t<n\), we separately consider the following  two cases.

\noindent\textbf{Case 1:} If \(\delta(f_{\sigma(s+1)})=\dots=\delta(f_{\sigma(t)})=+1\), 
let \(g=f_{\sigma(n-1)}\circ\dots\circ f_{\sigma(2)}\).
Then  Lemma \ref{lemma:gammaprec} and the optimality of \(\sigma\) imply  
\(\gamma(f_{\sigma(1)})\ge \gamma(g)\ge \gamma(f_{\sigma(n)})\), 
since \(-\delta(f_{\sigma(1)})=\delta(g)=-\delta(f_{\sigma(n)})=+1\).
This proves the lemma. 

\noindent\textbf{Case 2:} If \(\delta(f_{\sigma(s+1)})=\dots=\delta(f_{\sigma(t)})=-1\), 
then let \(h_1=f_{\sigma(s)}\circ\dots\circ f_{\sigma(1)}\),
\(h_2=f_{\sigma(t)}\circ\dots\circ f_{\sigma(s+1)}\) and
\(h_3=f_{\sigma(n)}\circ\dots\circ f_{\sigma(t+1)}\).
If \(\gamma(h_1)<\gamma(h_2)\), then \(h_3\circ h_2\circ h_1(x)\le h_3\circ h_1\circ h_2(x)\)
by \ref{lemma:gammaprec+-} in Lemma \ref{lemma:gammaprec}.
If \(\gamma(h_2)<\gamma(h_3)\), then \(h_3\circ h_2\circ h_1(x)\le h_2\circ h_3\circ h_1(x)\)
by \ref{lemma:gammaprec+-} in Lemma \ref{lemma:gammaprec}.
Otherwise (i.e., \(\gamma(h_1)\ge \gamma(h_2)\ge \gamma(h_3)\)), we have
\begin{align*}
h_3\circ h_2\circ h_1(x)\le \max\{h_2\circ h_1\circ h_3(x),\, h_1\circ h_3\circ h_2(x)\}
\end{align*}
by Lemma \ref{lemma:comp+-+}.
In either case, we can obtain a desired optimal solution by modifying $\sigma$. 
\end{proof}

By Lemmas \ref{lemma:prod>} and \ref{lemma:prod<}, we obtain Lemma \ref{lemma:tc}.

% Algorithm \ref{alg:tc} is also applicable for 
% the minimum total composition ordering problem \(((f_i)_{i\in[n]},c)\),
% since the optimal solution for this problem is 
% the same as the optimal solution  for
% maximum total composition ordering problem \(((\tilde{f}_i)_{i\in[n]},c)\)
% where \(\tilde{f}_i(x)=-f_i(-x)\) by Lemma \ref{lemma:minmax}.

\section{Negative Results}\label{sec:nr}
In the previous sections, we show that both the total and partial composition ordering problems can be solved efficiently if all $f_i$'s are monotone linear.  
It turns out that they cannot be generalized to nonlinear functions $f_i$. 
In this section, we show the optimal composition ordering problems  are in general intractable, even if all $f_i$'s are 
monotone increasing, piecewise linear functions with at most two pieces. 
We remark that the maximum total composition ordering problem is known to be \strongly NP-hard, even if all $f_i$'s are  
 monotone increasing,  
{\em concave},  piecewise linear functions with at most two pieces \cite{cheng1998tco}, which can be shown by considering the time-dependent scheduling problem. 
%In this paper, we consider the convex case for the maximum total composition ordering problem and the convex and concave cases for the partial composition ordering problem. 

For our reductions, 
we use the following NP-complete problems (see \cite{garey1979cai,ng2010ppa}).
\begin{description}
\item[\PARTITION:] Given $n$ positive integers $a_1,\dots,a_n$ with $\sum_{i=1}^n a_i=2T$,
ask whether exists a subset \(I\subseteq [n]\) such that $\sum_{i \in I}a_i=T$. 
\item[\PRODUCTPARTITION:] Given $n$ positive integers $a_1,\dots,a_n$ with $\prod_{i=1}^n a_i=T^2$,
ask whether there exists a subset \(I\subseteq [n]\) such that $\prod_{i\in I}a_i=T$. 
\end{description}
We use \PARTITION{} problem for concave case and
\PRODUCTPARTITION{} for convex case.
%Due to the space limitation, we provide the proof for convex case in Appendix.

\subsection{Monotone increasing, concave, piecewise linear functions with at most two pieces}
In this section, we consider the case in which all $f_i$'s are monotone increasing, concave,  piecewise linear functions with at most two pieces, 
that is, 
$f_i$ is given as
\begin{equation}
\label{eq-sec5-a1} 
f_i(x)=\min\{a_i^1x+b_i^1,\,a_i^2x+b_i^2\}
\end{equation}
for some reals $a_i^1$, $a_i^2$, $b_i^1$ and $b_i^2$ with $a_i^1, a_i^2>0$.

%% \begin{theorem}\label{theorem:2concaveNPH}
%% The maximum partial composition ordering problem is NP-hard, even if  all $f_i$'s are monotone increasing,  concave,  piecewise linear functions with at most two pieces. 
%% \end{theorem}
\renewcommand{\proofname}{Proof for Theorem \ref{theorem:hardness} (i)}
\begin{proof}
We show that \PARTITION{}  can be reduced to the problem. 
Let  $a_1,\dots, a_n$ denote positive integers with $\sum_{i=1}^n a_i=2T$.
We construct $n+2$ functions $f_i$ $(i=1, \dots , n+2)$ as follows:
\begin{align*}
f_i(x)=
\begin{cases}
x+a_i &{\rm if} \ i=1,\dots , n,\\
\min\left\{2x,~\frac{1}{2}x+\frac{3}{2}T\right\} &{\rm if} \ i=n+1,\\
6\alpha T(x-(3T-\frac{1}{2}))+(3T-\frac{1}{2}) &{\rm if} \ i=n+2.
\end{cases}
\end{align*}
It is clear  that all $f_i$'s are monotone, concave, and piecewise linear with at most two pieces.  
Moreover, we note that all $f_i$'s ($i=1,\dots,n+1$) satisfy $f_i(x) \geq x$ if $0\leq x \leq 3T$,
and \(f_{n+2}(x)\le x\) if $x \leq 3T-1/2$.
We claim that $3T$ is the optimal value for the maximum partial (total) composition ordering problem $((f_i)_{i\in [n+1]},c=0)$ 
if there exists a partition \(I\subseteq [n]\) such that $\sum_{i \in I}a_i=T$, 
and the optimal value is at most \(3T-1/2\) if $\sum_{i \in I}a_i\ne T$ for any partition \(I\subseteq [n]\).
This implies that
the optimal value for the maximum partial (total) composition ordering problem $((f_i)_{i\in[n+2]},c=0)$ is at least \(3\alpha T\) if $\sum_{i \in I}a_i=T$ for some \(I\subseteq [n]\), 
and at most \(3T\) if $\sum_{i \in I}a_i\ne T$ for any partition \(I\subseteq [n]\),
since \(f_{n+2}(3T)=3\alpha T+3T-1/2>3\alpha T\) and \(f_{n+2}(x)\le x\) if $x \le 3T-1/2$.
Thus, there exists no $\alpha$-approximation algorithm for the problems unless P=NP.

Let $\sigma:[n+1] \to [n+1]$ denote a permutation with $\sigma(l)=n+1$. 
Then define $I=\{\sigma(i) \mid i=1, \dots , l-1\}$ and 
$q=\sum_{i \in I} a_{i}$.
Note that $\sum_{i=l+1}^{n+1} a_{\sigma(i)}=\sum_{i \not\in I} a_i=2T-q$. 
Consider the function composition by $\sigma$:
\begin{align*}
f_{\sigma(n+1)}&\circ\dots\circ f_{\sigma(l+1)}\circ f_{\sigma(l)} \circ 
f_{\sigma(l-1)}\circ\dots\circ f_{\sigma(1)}(0)\\
&=f_{\sigma(n)}\circ\dots\circ f_{\sigma(l+1)}\circ f_{n+1}(q)\\
&=f_{\sigma(n)}\circ\dots\circ f_{\sigma(l+1)}\left(\min\left\{2q,~\frac{1}{2}q+\frac{3}{2}T\right\}\right)\\
&=\min\left\{2q,~\frac{1}{2}q+\frac{3}{2}T\right\}+2T-q \, =\min\left\{q,~-\frac{1}{2}q+\frac{3}{2}T\right\}+2T.
\end{align*}
Note that $\min\left\{q,~-\frac{1}{2}q+\frac{3}{2}T\right\} \leq T$ holds, where the equality holds only when $q=T$. 
This implies that 
\begin{equation}
\label{eq-sec5-a3}
f_{\sigma(n+1)}\circ\dots\circ f_{\sigma(l+1)}\circ f_{\sigma(l)} \circ 
f_{\sigma(l-1)}\circ\dots\circ f_{\sigma(1)}(0) 
\begin{cases}
= 3T &(q=T),\\
\leq 3T-1/2 &(q\ne T)
\end{cases}
\end{equation}
since \(q\) is an integer, which proves the claim. 
%% Thus, $I$ above is a desired partition if and only if 
%%  the total function composition by $\sigma$ has value $3T$.  
%% Recall that  all $f_i$ satisfy $f_i \geq x$ if $0\leq x \leq 3T$.
%% By this together with (\ref{eq-sec5-a3}) and Lemma \ref{lemma:partialtotal}, 
%%  $k$ can be restricted to $k=n+1$ for the instance of 
%% the maximum partial composition ordering problem, 
%% i.e., the partial setting can be solved by dealing with the total setting.
\end{proof}
\renewcommand{\proofname}{proof}

By Lemma \ref{lemma:partialtotal}, we have the following corollary. 
We also have the following corollary.
\begin{corollary}\label{cor:2concaveNPH}
The maximum total composition ordering problem 
is NP-hard, even if all $f_i$'s are represented by  $f_i(x)=\max\{x,\min\{a_i^1x+b_i^1,a_i^2x+b_i^2\}\}$ for some 
 reals $a_i^1$, $a_i^2$, $b_i^1$ and $b_i^2$ with $a_i^1, a_i^2>0$. 
\end{corollary}

\subsection{Monotone increasing,  convex,  piecewise linear functions with at most two pieces}
In this section, we consider the case in which all $f_i$'s are monotone increasing,  convex,  piecewise linear functions with at most two pieces, 
that is, 
$f_i$ is given as
\begin{equation}
\label{eq-sec5-a4} 
f_i(x)=\max\{a_i^1x+b_i^1,\,a_i^2x+b_i^2\}
\end{equation}
for some reals $a_i^1$, $a_i^2$, $b_i^1$ and $b_i^2$ with $a_i^1, a_i^2>0$. 
Before showing the intractability of the problems, we present two basic properties for the function composition. 

For an integer $i \in [n]$, let $g_i=a_i(x-d)+d$. 
Then we have 
\begin{equation}
\label{eq-sec5-a5}
g_{n}\circ g_{n-1}\circ \dots\circ g_{1}(x)=(x-d)\prod_{i=1}^na_i+d.  
\end{equation}
Thus, $\prod_{i=1}^na_i >0$ implies the following inequalities:
\begin{equation}
\label{eq-sec5-a6}
\left.
\begin{array}{ll}
g_{n}\circ g_{n-1}\circ \dots\circ g_{1}(x) < d & {\rm if } \ x< d,\\
g_{n}\circ g_{n-1}\circ \dots\circ g_{1}(x) = d & {\rm if } \ x= d,\\
g_{n}\circ g_{n-1}\circ \dots\circ g_{1}(x)> d & {\rm if } \  x> d.
\end{array}
\right.
\end{equation}

We are now ready to prove the intractability. 
%% \begin{theorem}\label{theorem:2convexNPH}
%% Both the maximum total and partial composition ordering problems are NP-hard,  
%% even if  all $f_i$'s are monotone increasing,  convex,  piecewise linear functions with at most two pieces. 
%% \end{theorem}
\renewcommand{\proofname}{Proof for Theorem \ref{theorem:hardness} (ii)}
\begin{proof}
We show that \PRODUCTPARTITION{} can be reduced to them. 

Let  $a_1,\dots, a_n\,(> 1)$ denote positive integers with $\prod_{i=1}^na_i=T^2$.
We construct $n+2$ functions $f_i$ $(i=1, \dots , n+2)$ as follows:
\begin{align*}
f_i(x)=
\begin{cases}
\max\left\{\frac{1}{a_i}(x-T^2)+T^2,~ a_i(x-T^2)+T^2\right\} &\text{if}  \ i=1,\dots , n,\\
x+2T &\text{if} \ i=n+1,\\
4\alpha (T+1)^2\left(x-2T^2+\left(\frac{T}{T+1}\right)^2\right)-2T^2+\left(\frac{T}{T+1}\right)^2
&\text{if} \ i=n+2,\\
\end{cases}
\end{align*}

\noindent 
It is clear that all $f_i$'s are monotone, convex, and piecewise linear with at most two pieces.  
Moreover, we note that  $f_i \geq x$ holds for all functions $f_i$, which together with Lemma \ref{lemma:partialtotal} implies that  
the maximum partial and total composition ordering problems are equivalent for the functions $f_i$. 
Therefore, we deal with the total setting only. 
We now claim that $2T^2$ is the optimal value for the maximum partial composition ordering problem $((f_i)_{i\in [n+1},c=0)$ 
if there exists a desired partition \(I\subseteq [n]\) for \PRODUCTPARTITION{}, i.e., $\prod_{i\in I}a_i=T$,
and at most \(2T^2-(T/(T+1))^2\) otherwise.
This implies that the optimal value for the maximum total composition ordering problem $((f_i)_{i\in[n+2]},c=0)$ is at least \(2\alpha T^2\) 
if $\prod_{i\in I}a_i=T$ for an \(I\subseteq [n]\),
and at most \(2T^2\) if $\prod_{i\in I}a_i\ne T$ for any \(I\subseteq [n]\),
since \(f_{n+2}(2T^2)>2\alpha T^2\) and \(f_{n+2}(x)\le x\) if \(x\le 2T^2-(T/(T+1))^2\).
Thus, there exists no \(\alpha\)-approximation algorithm for the problems unless P=NP.

Let $\sigma:[n+1] \to [n+1]$ denote a permutation with $\sigma(l)=n+1$. 
Then define $I=\{\sigma(i) \mid i=1, \dots , l-1\}$ and 
$p=\frac{1}{\prod_{i \in I} a_{i}}$.
Note that $\prod_{i=l+1}^{n+1} a_{\sigma(i)}=\prod_{i \not\in I} a_i=pT^2$. 
Consider the function composition by $\sigma$:
\begin{align}
f_{\sigma(n+1)} \circ\dots \circ f_{\sigma(l+1)}\circ&  f_{\sigma(l)} \circ   f_{\sigma(l-1)}\circ\dots\circ f_{\sigma(1)}(0)\nonumber\\
&= \,\, f_{\sigma(n)}\circ\dots\circ f_{\sigma(l+1)}\circ f_{n+1}(T^2(1-p)) \label{eq-sec5-a7}\\
&=\,\, f_{\sigma(n)}\circ\dots\circ f_{\sigma(l+1)}(T^2(1-p)+2T) \nonumber\\
&\leq \,\, pT^2 (T^2(1-p)+2T-T^2)+T^2 \label{eq-sec5-a8}\\
&=\,\, 2T^2 -T^2(pT-1)^2 \nonumber
% &\leq \,\, 2T^2,  \label{eq-sec5-a9}
\end{align}
where \eqref{eq-sec5-a7} follows from \eqref{eq-sec5-a5} and \eqref{eq-sec5-a6}, 
and  
(\ref{eq-sec5-a8}) follows from  \eqref{eq-sec5-a5} and $a_{\sigma(i)} > 1$ for all $\ i \geq l+1$. 
We also note that \eqref{eq-sec5-a8} is satisfied by equality if and only if  \(T^2(1-p)+2T\ge T^2\),
i.e., \(p\le 2/T\).
%% Therefore, the value of the total composition is at most \(2T^2\) for any permutation $\sigma$, and it  
%% becomes \(2T^2\) only when $p=1/T$.
Thus, we have
\begin{align*}
f_{\sigma(n+1)}\circ\dots\circ f_{\sigma(1)}(0)
\begin{cases}
=2T^2 &(p=1/T),\\
\le 2T^2-\left(\frac{T}{T+1}\right)^2&(p\ne 1/T)
\end{cases}
\end{align*}
since $1/p$ is an integer, which proves the claim.
\end{proof}
\renewcommand{\proofname}{proof}

By Lemma~\ref{lemma:partialtotal}, we also have the following result. 
\begin{corollary}
The maximum total composition ordering problem 
is NP-hard, even if all $f_i$'s are represented by  $f_i(x)=\max\{x, a_i^1x+b_i^1,a_i^2x+b_i^2\}$ for some 
 reals $a_i^1$, $a_i^2$, $b_i^1$ and $b_i^2$ with $a_i^1, a_i^2>0$. 
\end{corollary}

\bibliographystyle{plain}
\newpage
\bibliography{all}

\begin{thebibliography}{10}

\bibitem{babaioff2007aks}
Moshe Babaioff, Nicole Immorlica, David Kempe, and Robert Kleinberg.
\newblock A knapsack secretary problem with applications.
\newblock {\em Approximation, Randomization, and Combinatorial Optimization.
  Algorithms and Techniques}, pages 16--28, 2007.

\bibitem{babaioff2007msp}
Moshe Babaioff, Nicole Immorlica, and Robert Kleinberg.
\newblock Matroids, secretary problems, and online mechanisms.
\newblock In {\em Proceedings of the eighteenth annual ACM-SIAM symposium on
  Discrete algorithms}, pages 434--443, 2007.

\bibitem{cai1998oas}
Jin-Yi Cai, Pu~Cai, and Yixin Zhu.
\newblock On a scheduling problem of time deteriorating jobs.
\newblock {\em Journal of Complexity}, 14(2):190--209, 1998.

\bibitem{cheng1998tco}
T.~C.~E. Cheng and Q.~Ding.
\newblock The complexity of scheduling starting time dependent tasks with
  release times.
\newblock {\em Information Processing Letters}, 65(2):75--79, 1998.

\bibitem{cheng2004acs}
T.~C.~E. Cheng, Q~Ding, and B.M.T Lin.
\newblock A concise survey of scheduling with time-dependent processing times.
\newblock {\em European Journal of Operational Research}, 152(1):1--13, 2004.

\bibitem{cheng2003sjw}
T.~C.~E. Cheng, Qing Ding, Mikhail~Y. Kovalyov, Aleksander Bachman, and Adam
  Janiak.
\newblock Scheduling jobs with piecewise linear decreasing processing times.
\newblock {\em Naval Research Logistics}, 50(6):531--554, 2003.

\bibitem{dean2005aaa}
B.C. Dean, M.X. Goemans, and J.~Vondr{\'a}k.
\newblock Adaptivity and approximation for stochastic packing problems.
\newblock In {\em Proceedings of the sixteenth annual ACM-SIAM Symposium on
  Discrete Algorithms}, pages 395--404. Society for Industrial and Applied
  Mathematics, 2005.

\bibitem{dean2008ats}
B.C. Dean, M.X. Goemans, and J.~Vondr{\'a}k.
\newblock Approximating the stochastic knapsack problem: the benefit of
  adaptivity.
\newblock {\em Mathematics of Operations Research}, 33(4):945--964, 2008.

\bibitem{ferguson1989wst}
Thomas~S. Ferguson.
\newblock Who solved the secretary problem?
\newblock {\em Statical Science}, 4(3):282--289, 1989.

\bibitem{garey1979cai}
Michael~R. Garey and David~S. Johnson.
\newblock {\em Computers and Intractability: A Guide to the Theory of
  {NP}-Completeness}.
\newblock Freeman New York, 1979.

\bibitem{gawiejnowicz2007sdj}
S.~Gawiejnowicz.
\newblock Scheduling deteriorating jobs subject to job or machine availability
  constraints.
\newblock {\em European Journal of Operational Research}, 180(1):472--478,
  2007.

\bibitem{gawiejnowicz2008tds}
Stanis{\l}aw Gawiejnowicz.
\newblock {\em Time-Dependent Scheduling}.
\newblock Springer, 2008.

\bibitem{gawiejnowicz1995sjw}
Stanis{\l}aw Gawiejnowicz and Lidia Pankowska.
\newblock Scheduling jobs with varying processing times.
\newblock {\em Information Processing Letters}, 54(3):175--178, 1995.

\bibitem{gupta1988sfs}
Jatinder~N.D. Gupta and Sushil~K. Gupta.
\newblock Single facility scheduling with nonlinear processing times.
\newblock {\em Computers \& Industrial Engineering}, 14(4):387--393, 1988.

\bibitem{ho1993cos}
Kevin I-J. Ho, Joseph Y-T. Leung, and W-D. Wei.
\newblock Complexity of scheduling tasks with time-dependent execution times.
\newblock {\em Information Processing Letters}, 48(6):315--320, 1993.

\bibitem{melnikov1980ppo}
O.~I. Melnikov and Y.~M. Shafransky.
\newblock Parametric problem of scheduling theory.
\newblock {\em Cybernetics}, 15:352--357, 1980.

\bibitem{mosheiov1994sju}
Gur Mosheiov.
\newblock Scheduling jobs under simple linear deterioration.
\newblock {\em Computers \& Operations Research}, 21(6):653--659, 1994.

\bibitem{ng2010ppa}
C.~T. Ng, M.S. Barketau, T.~C.~E. Cheng, and Mikhail~Y. Kovalyov.
\newblock ``{P}roduct partition'' and related problems of scheduling and
  systems reliability: Computational complexity and approximation.
\newblock {\em European Journal of Operational Research}, 207:601--604, 2010.

\bibitem{gharan2011ovo}
Shayan Oveis~Gharan and Jan Vondr\'{a}k.
\newblock On variants of the matroid secretary problem.
\newblock In {\em Proceedings of the 19th Annual European Symposium on
  Algorithms}, pages 335--346, 2011.

\bibitem{tanaev1994sts}
V.~S. Tanaev, V.~S. Gordon, and Y.~M. Shafransky.
\newblock {\em Scheduling Theory: Single-Stage Systems}.
\newblock Kluwer Academic Publishers, 1994.

\bibitem{wajs1986paf}
W.~Wajs.
\newblock Polynomial algorithm for dynamic sequencing problem.
\newblock {\em Archiwum Automatyki i Telemechaniki}, 31(3):209--213, 1986.

\end{thebibliography}
\newpage
\appendix
\setlength\intextsep{10pt}
\section*{Appendix: Omitted Proofs}
\subsection*{Proof of Lemma \ref{lemma:gammaprec}}
\newtheorem*{lemma:gammaprec}{Lemma \ref{lemma:gammaprec}}
\begin{lemma:gammaprec}
Let \(f_i(x)=a_ix+b_i\) and \(f_j(x)=a_jx+b_j\) be $($non-identity$)$ monotone nondecreasing functions $($i.e., \((a_i,b_i),(a_j,b_j)\ne (1,0)\), \(a_i,a_j\ge 0\)$)$. 
Then we have the following statements;
\begin{enumerate}
\item if \(a_i,a_j=1\), then \(f_i\simeq f_j\),
\item if \(a_i,a_j\ge 1\) and \(a_i\cdot a_j>1\), then \(f_i\preceq f_j\siff\gamma(f_i)\le \gamma(f_j)\),
\item if \(a_i,a_j<1\), then \(f_i\preceq f_j\siff\gamma(f_i)\le \gamma(f_j)\),
\item if \(a_i\ge 1\), \(a_j<1\), then \(f_i\preceq f_j\siff\gamma(f_i)\ge \gamma(f_j)\) and \(f_i\succeq f_j\siff\gamma(f_i)\le \gamma(f_j)\).
\end{enumerate}
\end{lemma:gammaprec}
\begin{proof}
$\boldsymbol{ (a)}$:  It immediately follows from  \(f_i\circ f_j(x)=f_j\circ f_i(x)=x+b_i+b_j\).

\smallskip
\noindent
$\boldsymbol{(b)}$: If \(a_i,a_j>1\),  then the lemma holds, since we have the following equivalences  \(\eqref{eq:lprec}\siff \frac{b_i}{1-a_i}\le \frac{b_j}{1-a_j}\siff \gamma(f_i)\le\gamma(f_j)\). 
If \(a_i>1\) and \(a_j=1\), 
 then the lemma holds, since we have the following equivalences \(\eqref{eq:lprec}\siff 0\le b_j(1-a_i)\siff b_j<0\siff \gamma(f_j)=+\infty\siff \gamma(f_i)\le\gamma(f_j)\).
Otherwise (i.e., \(a_i=1\) and \(a_j>1\)), 
we have \(\eqref{eq:lprec}\siff b_i(1-a_j)\le 0\siff b_i>0\siff \gamma(f_i)=-\infty\siff \gamma(f_i)\le\gamma(f_j)\), which prove the lemma.

\smallskip
\noindent
$\boldsymbol{(c)}$: The lemma holds, since we have the following equivalences  \(\eqref{eq:lprec}\siff \frac{b_i}{1-a_i}\le \frac{b_j}{1-a_j}\siff \gamma(f_i)\le\gamma(f_j)\).

\smallskip
\noindent
$\boldsymbol{(d)}$: If \(a_i>1\), the lemma holds since we have the following equivalences \(\eqref{eq:lprec}\siff \frac{b_i}{1-a_i}\ge \frac{b_j}{1-a_j}\siff \gamma(f_i)\ge\gamma(f_j)\).
On the other hand, if \(a_i=1\), then \(f_i\preceq f_j\siff b_i(1-a_j)\le 0\siff b_i<0\siff \gamma(f_i)=+\infty\siff \gamma(f_i)\ge\gamma(f_j)\),
and \(f_i\succeq f_j\siff b_i(1-a_j)\ge 0\siff b_i>0\siff \gamma(f_i)=-\infty\siff \gamma(f_i)\le\gamma(f_j)\).
% \item If \(a_j>1\), the lemma holds since the equation \(\eqref{eq:lprec}\siff \frac{b_i}{1-a_i}\le \frac{b_j}{1-a_j}\siff \gamma(f_i)\ge\gamma(f_j)\).
% On the other hand, if  \(a_j=1\), then \(f_i\preceq f_j\siff -b_j(a_i-1)\ge 0\siff b_j>0\siff \gamma(f_j)=-\infty\siff \gamma(f_i)\ge\gamma(f_j)\).
\end{proof}

\subsection*{Proof of Lemma \ref{lemma:gammaprec-check}}
\newtheorem*{lemma:gammaprec-check}{Lemma \ref{lemma:gammaprec-check}}
\begin{lemma:gammaprec-check}
For $($non-identity$)$ monotone nondecreasing linear functions \(f_i(x)=a_ix+b_i\) and \(f_j(x)=a_jx+b_j\),  
we have the following statements;
\begin{enumerate}
\item if \(a_i, a_j \geq 1\) and \(\gamma(f_i)\le \gamma(f_j)\), then \(\overline{f}_i\preceq \overline{f}_j\),
\item if \(a_i, a_j < 1\)  and \(\gamma(f_i)\le \gamma(f_j)\), then \(\overline{f}_i\preceq \overline{f}_j\),
\item if \(a_i <1 \), \(a_j \geq 1\), and \(\gamma(f_i)\le \gamma(f_j)\), then \(\overline{f}_i\simeq \overline{f}_j\),
\item if \(a_i \geq 1\), \(a_j<1\), and \(\gamma(f_i)\le \gamma(f_j)\), then \(\overline{f}_i\succeq \overline{f}_j\).
\end{enumerate}
\end{lemma:gammaprec-check}
\begin{proof}
\noindent
$\boldsymbol{(a)}$:  We prove that \(\overline{f}_j\circ\overline{f}_i(x)\ge \overline{f}_i\circ\overline{f}_j(x)\) holds for any \(x\).
We separately consider three cases \(x< \gamma(f_i)\), \(\gamma(f_i)\le x\le \gamma(f_j)\), and \(\gamma(f_j)< x\) 
(see Figure  \ref{fig:gammaprec-check} \subref{fig:gammaprec-check++}).

\smallskip
\noindent \textbf{Case} $\boldsymbol{a}${\bf -1}: If \(x< \gamma(f_i)\), then 
we have
\(\overline{f}_i\circ \overline{f}_j(x)=\overline{f}_i(x)=x\) 
and \(\overline{f}_j\circ \overline{f}_i(x)=\overline{f}_j(x)=x\)
by \(x<\gamma(f_i)\le\gamma(f_j)\).
Thus, we obtain \(\overline{f}_j\circ \overline{f}_i(x)= \overline{f}_i\circ \overline{f}_j(x)\).

\smallskip
\noindent \textbf{Case} $\boldsymbol{a}${\bf -2}: If \(\gamma(f_i)\le x\le \gamma(f_j)\),  then it holds that 
\(\overline{f}_i\circ \overline{f}_j(x)=\overline{f}_i(x)=f_i(x)\) and 
\(\overline{f}_j\circ \overline{f}_i(x)=\overline{f}_j(f_i(x))\) 
by \(\gamma(f_i)\le x\le \gamma(f_j)\).
Thus, we obtain \(\overline{f}_j\circ \overline{f}_i(x)\ge \overline{f}_i\circ \overline{f}_j(x)\), 
since \(\overline{f}_j(y)\ge y\) for any \(y\).

\smallskip
\noindent\textbf{Case} $\boldsymbol{a}${\bf -3}: If \(\gamma(f_j)< x\), then 
we have
\(\overline{f}_i\circ \overline{f}_j(x)=\overline{f}_i(f_j(x))=f_i(f_j(x))\) 
by \(\gamma(f_i)\le\gamma(f_j)<x\le f_j(x)\),
and \(\overline{f}_j\circ \overline{f}_i(x)=\overline{f}_j(f_i(x))=f_j(f_i(x))\)
by \(\gamma(f_i)\le\gamma(f_j)<x\le f_i(x)\).
Thus, we obtain \(\overline{f}_j\circ \overline{f}_i(x)\ge \overline{f}_i\circ \overline{f}_j(x)\)
by \ref{lemma:gammaprec11} and \ref{lemma:gammaprec++} in Lemma~\ref{lemma:gammaprec}.

\smallskip
\noindent
$\boldsymbol{(b)}$:  We prove that \(\overline{f}_j\circ\overline{f}_i(x)\ge \overline{f}_i\circ\overline{f}_j(x)\) holds for any \(x\).
We separately consider four cases \(x< f_j^{-1}(\gamma(f_i))\), \(f_j^{-1}(\gamma(f_i))\le x< \gamma(f_i)\),
\(\gamma(f_i)\le x< \gamma(f_j)\), and \(\gamma(f_j)\le x\)
%where we define \(f_j^{-1}(\gamma(f_i))=-\infty\) 
(see Figure \ref{fig:gammaprec-check} \subref{fig:gammaprec-check--}).

\smallskip
\noindent\textbf{Case} $\boldsymbol{b}${\bf -1}: If \(x< f_j^{-1}(\gamma(f_i))\), then 
we have 
\(\overline{f}_i\circ \overline{f}_j(x)=\overline{f}_i(f_j(x))=f_i(f_j(x))\) 
by \(x\le f_j(x)\le\gamma(f_i)\le \gamma(f_j)\),
and \(\overline{f}_j\circ \overline{f}_i(x)=\overline{f}_j(f_i(x))=f_j(f_i(x))\)
by \(x\le f_i(x)\le\gamma(f_i)\le\gamma(f_j)\).
Thus, we obtain \(\overline{f}_j\circ \overline{f}_i(x)\ge\overline{f}_i\circ \overline{f}_j(x)\)
by \ref{lemma:gammaprec--} in Lemma \ref{lemma:gammaprec}.

\smallskip
\noindent\textbf{Case} $\boldsymbol{b}${\bf -2}: If \(f_j^{-1}(\gamma(f_i))\le x< \gamma(f_i)\), then 
we have 
\(\overline{f}_i\circ \overline{f}_j(x)=\overline{f}_i(f_j(x))=f_j(x)\)
and \(\overline{f}_j\circ \overline{f}_i(x)=\overline{f}_j(f_i(x))=f_j(f_i(x))\)
by \(x\le f_i(x)\le\gamma(f_i)\le f_j(x)\le \gamma(f_j)\).
Thus, we obtain \(\overline{f}_j\circ \overline{f}_i(x)\ge \overline{f}_i\circ \overline{f}_j(x)\), 
since \(f_i(x)\ge x\) and \(f_j\) is monotone nondecreasing.

\smallskip
\noindent\textbf{Case} $\boldsymbol{b}${\bf -3}: If \(\gamma(f_i)\le x< \gamma(f_j)\), 
then we have 
\(\overline{f}_i\circ \overline{f}_j(x)=\overline{f}_i(f_j(x))=f_j(x)\) 
and \(\overline{f}_j\circ \overline{f}_i(x)=\overline{f}_j(x)=f_j(x)\)
by \(\gamma(f_i)\le x\le f_j(x)< \gamma(f_j)\).
Thus, we obtain \(\overline{f}_j\circ \overline{f}_i(x) = \overline{f}_i\circ \overline{f}_j(x)\).

\smallskip
\noindent\textbf{Case} $\boldsymbol{b}${\bf -4}: If \(\gamma(f_j)\le x\), 
then we have 
\(\overline{f}_i\circ \overline{f}_j(x)=\overline{f}_i(x)=x\) 
and \(\overline{f}_j\circ \overline{f}_i(x)=\overline{f}_j(x)=x\)
by \(\gamma(f_i)\le \gamma(f_j)\le x\).
Thus, we obtain \(\overline{f}_j\circ \overline{f}_i(x)= \overline{f}_i\circ \overline{f}_j(x)\).

\smallskip

\noindent
$\boldsymbol{(c)}$: 
We prove that \(\overline{f}_j\circ\overline{f}_i(x)= \overline{f}_i\circ\overline{f}_j(x)\) holds for any \(x\).
We separately consider three cases \(x< \gamma(f_i)\), \(\gamma(f_i)\le x< \gamma(f_j)\), and \(\gamma(f_j)\le x\)
(see Figure \ref{fig:gammaprec-check} \subref{fig:gammaprec-check-+}).

\smallskip
\noindent\textbf{Case} $\boldsymbol{c}${\bf -1}: If \(x<\gamma(f_i)\), 
then we  have 
\(\overline{f}_i\circ \overline{f}_j(x)=\overline{f}_i(x)=f_i(x)\) 
and \(\overline{f}_j\circ \overline{f}_i(x)=\overline{f}_j(f_i(x))=f_i(x)\)
by \(x\le f_i(x)\le\gamma(f_i)\le\gamma(f_j)\).
Thus, we obtain \(\overline{f}_j\circ \overline{f}_i(x)= \overline{f}_i\circ \overline{f}_j(x)\).

\smallskip
\noindent\textbf{Case} $\boldsymbol{c}${\bf -2}: If \(\gamma(f_i)\le x<\gamma(f_j)\), then 
we have 
\(\overline{f}_i\circ \overline{f}_j(x)=\overline{f}_i(x)=x\) 
and \(\overline{f}_j\circ \overline{f}_i(x)=\overline{f}_j(x)=x\)
by \(\gamma(f_i)\le x<\gamma(f_j)\).
Thus, we obtain \(\overline{f}_j\circ \overline{f}_i(x)= \overline{f}_i\circ \overline{f}_j(x)\).

\smallskip
\noindent\textbf{Case} $\boldsymbol{c}${\bf -3}: If \(\gamma(f_j)\le x\), 
then we have 
\(\overline{f}_i\circ \overline{f}_j(x)=\overline{f}_i(f_j(x))=f_j(x)\) 
and \(\overline{f}_j\circ \overline{f}_i(x)=\overline{f}_j(x)=f_j(x)\)
by \(\gamma(f_i)\le \gamma(f_j)\le x\le f_j(x)\).
Thus, we obtain \(\overline{f}_j\circ \overline{f}_i(x)= \overline{f}_i\circ \overline{f}_j(x)\).

\smallskip

\noindent
$\boldsymbol{(d)}$: 
We prove that \(\overline{f}_j\circ\overline{f}_i(x)\le \overline{f}_i\circ\overline{f}_j(x)\) holds for any \(x\).
We separately consider four cases \(x<\gamma(f_i)\), \(\gamma(f_i)\le x< f_i^{-1}(\gamma(f_j))\), \(f_i^{-1}(\gamma(f_j))\le x<\gamma(f_j)\), and \(\gamma(f_j)\le x\)
(see Figure \ref{fig:gammaprec-check} \subref{fig:gammaprec-check+-}).

\smallskip
\noindent\textbf{Case} $\boldsymbol{d}${\bf -1}: If \(x<\gamma(f_i)\), then 
we have 
\(\overline{f}_i\circ \overline{f}_j(x)=\overline{f}_i(f_j(x))\) 
and \(\overline{f}_j\circ \overline{f}_i(x)=\overline{f}_j(x)=f_j(x)\)
by \(x<\gamma(f_i)\le\gamma(f_j)\).
Thus, we obtain \(\overline{f}_j\circ \overline{f}_i(x)\le \overline{f}_i\circ \overline{f}_j(x)\), 
since \(\overline{f}_i(y)\ge y\) for any \(y\).

\smallskip
\noindent\textbf{Case} $\boldsymbol{d}${\bf -2}: If \(\gamma(f_i)\le x< f_i^{-1}(\gamma(f_j))\), then 
we have 
\(\overline{f}_i\circ \overline{f}_j(x)=\overline{f}_i(f_j(x))=f_i(f_j(x))\)
by \(\gamma(f_i)\le x\le f_j(x)\le \gamma(f_j)\),
and \(\overline{f}_j\circ \overline{f}_i(x)=\overline{f}_j(f_i(x))=f_j(f_i(x))\)
by \(\gamma(f_i)\le x\le f_i(x)\le \gamma(f_j)\).
Thus, we obtain \(\overline{f}_j\circ \overline{f}_i(x)\le \overline{f}_i\circ \overline{f}_j(x)\)
by \ref{lemma:gammaprec+-} in Lemma \ref{lemma:gammaprec}.

\smallskip
\noindent\textbf{Case} $\boldsymbol{d}${\bf -3}: If \(f_i^{-1}(\gamma(f_j))\le x< \gamma(f_j)\), then 
we have 
\(\overline{f}_i\circ \overline{f}_j(x)=\overline{f}_i(f_j(x))=f_i(f_j(x))\)
by \(\gamma(f_i)\le f_i^{-1}(\gamma(f_j))\le x\le f_j(x)\le\gamma(f_j)\)
and \(\overline{f}_j\circ \overline{f}_i(x)=\overline{f}_j(f_i(x))=f_i(x)\)
by \(\gamma(f_i)\le f_i^{-1}(\gamma(f_j))\le x\le\gamma(f_j)\le f_i(x)\).
Thus, we obtain \(\overline{f}_j\circ \overline{f}_i(x)\le \overline{f}_i\circ \overline{f}_j(x)\), 
since \(f_j(x)\ge x\) and \(f_i\) is monotone nondecreasing.

\smallskip
\noindent\textbf{Case} $\boldsymbol{d}${\bf -4}:  If \(\gamma(f_j)\le x\), then 
we have 
\(\overline{f}_i\circ \overline{f}_j(x)=\overline{f}_i(x)=f_i(x)\) 
and \(\overline{f}_j\circ \overline{f}_i(x)=\overline{f}_j(f_i(x))=f_i(x)\)
by \(\gamma(f_i)\le \gamma(f_j)\le x\le f_i(x)\).
Thus, we obtain \(\overline{f}_j\circ \overline{f}_i(x)= \overline{f}_i\circ \overline{f}_j(x)\).

\begin{figure}[t]
\centering{
\subfloat[\(a_i,a_j\ge 1,~ \gamma(f_i)\le \gamma(f_j)\)]%
{\includegraphics[width=0.30 \textwidth]{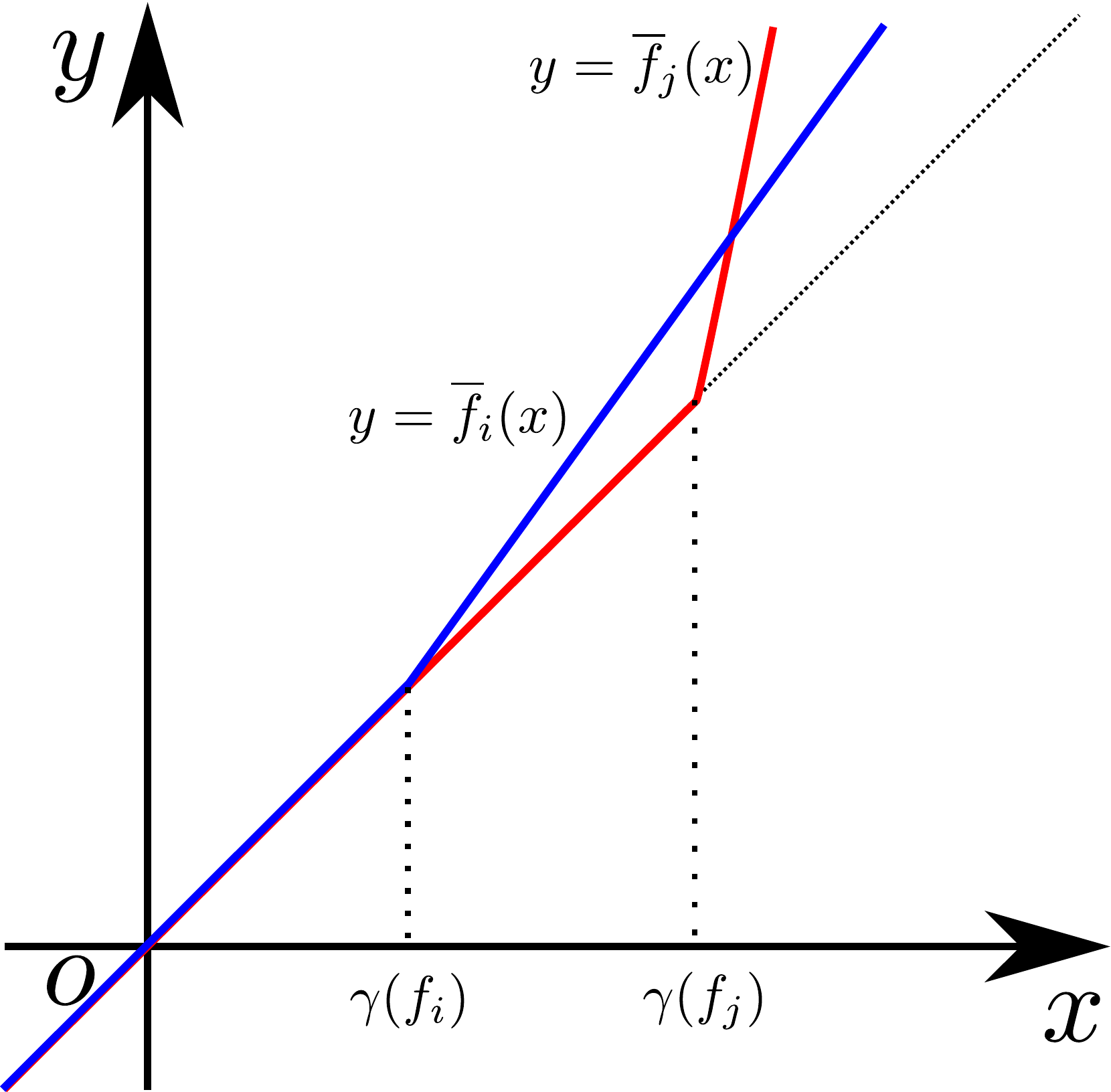}\label{fig:gammaprec-check++}}\qquad\qquad
\subfloat[\(0\le a_i,a_j< 1,~ \gamma(f_i)\le \gamma(f_j)\)]%
{\includegraphics[width=0.30 \textwidth]{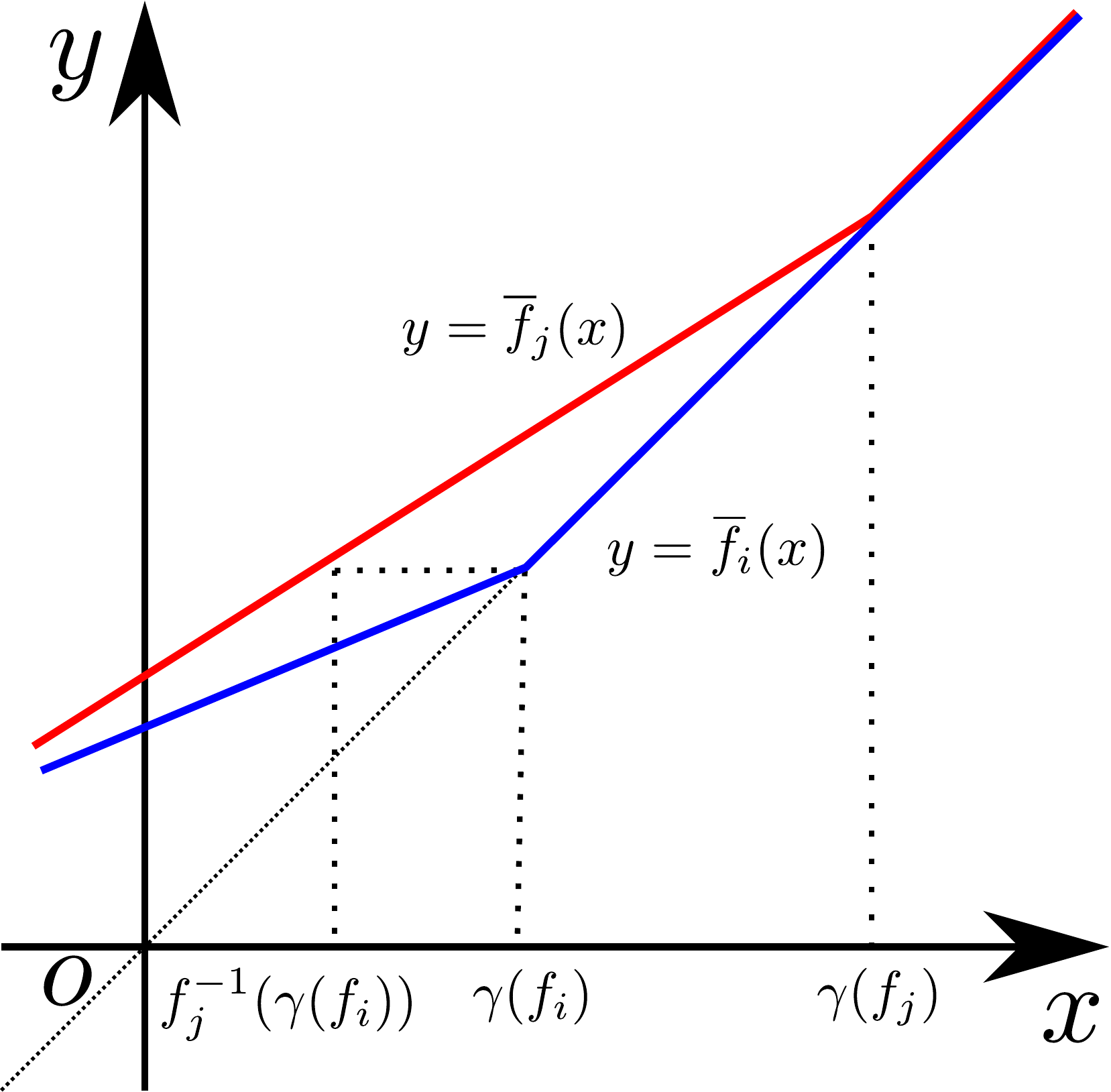}\label{fig:gammaprec-check--}}\\
\subfloat[\(0\le a_i<1,a_j\ge 1,~ \gamma(f_i)\le \gamma(f_j)\)]%
{\includegraphics[width=0.30 \textwidth]{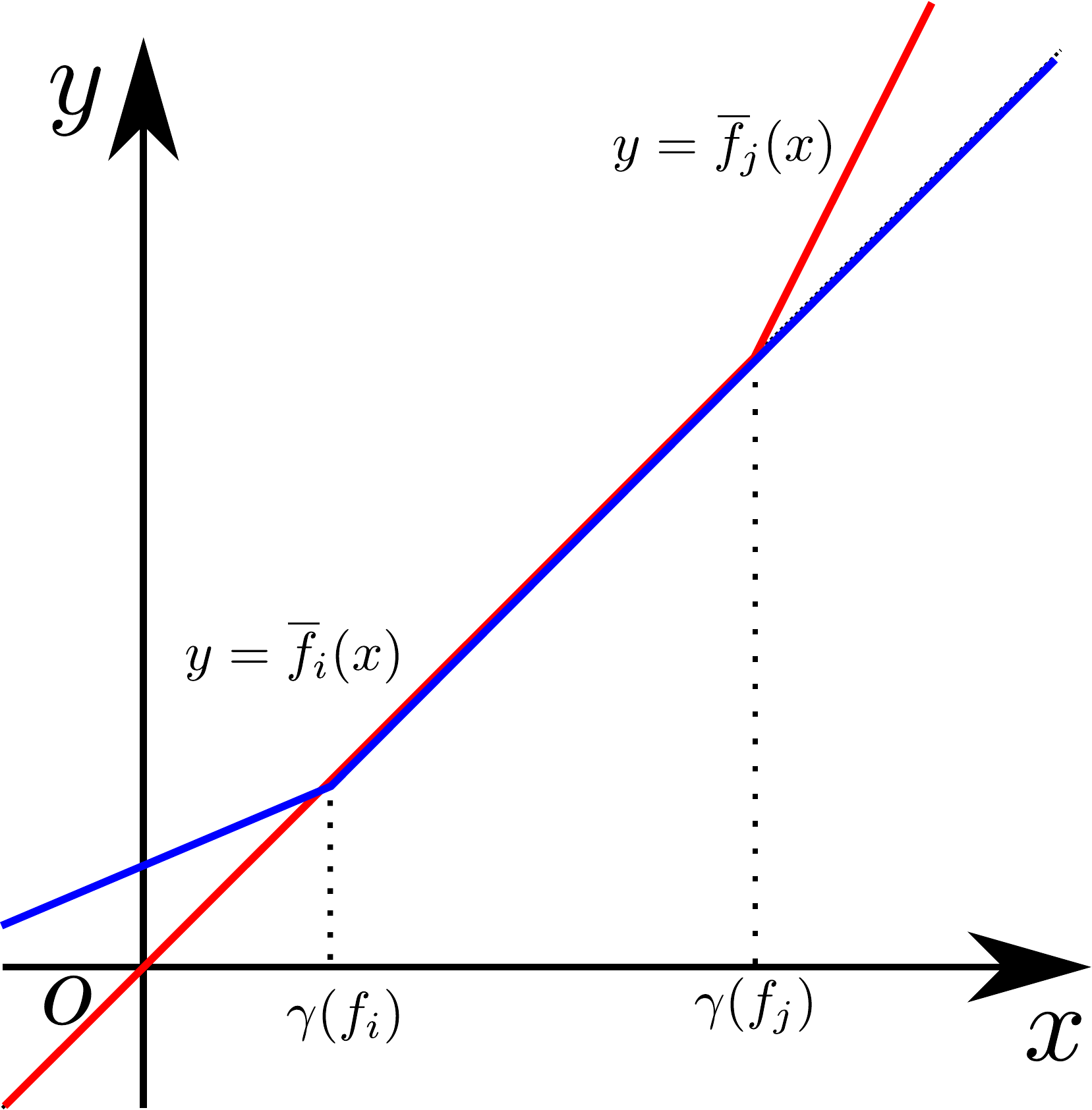}\label{fig:gammaprec-check-+}}\qquad\qquad
\subfloat[\(a_i\ge 1,~0\le a_j< 1,~ \gamma(f_i)\le \gamma(f_j)\)]%
{\includegraphics[width=0.30 \textwidth]{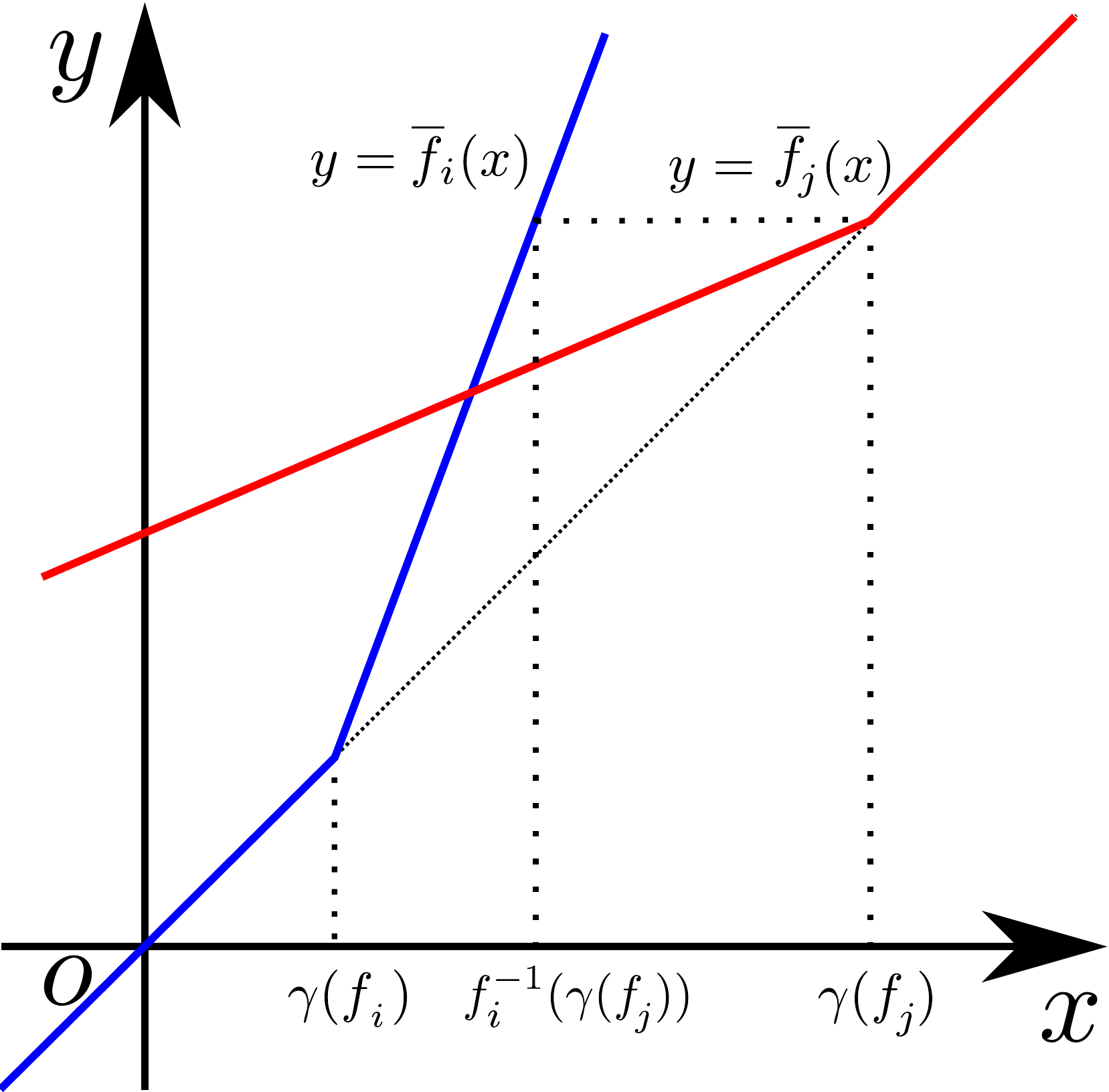}\label{fig:gammaprec-check+-}}%Caution: file name is inverse
}
\caption{Typical situations for the functions \(\overline{f}_i\) and \(\overline{f}_j\).} \label{fig:gammaprec-check}
\end{figure}
\end{proof}

\subsection*{Proof of Lemma \ref{lemma:gammacomp}}
\newtheorem*{lemma:gammacomp}{Lemma \ref{lemma:gammacomp}}
\begin{lemma:gammacomp}
For monotone nondecreasing linear functions \(f_i(x)=a_ix+b_i\) and \(f_j(x)=a_j x+b_j\) \((a_i,a_j\ge 0)\),
we have the following statements. 
\begin{enumerate}
\item If \(\gamma(f_i)=\gamma(f_j)\), then \(\gamma(f_i)=\gamma(f_j)= \gamma(f_j\circ f_i)\),
\item If \(\gamma(f_i)<\gamma(f_j)\) and \(a_i,a_j\ge 1\), then  \(\gamma(f_i)\le \gamma(f_j\circ f_i)\le\gamma(f_j)\),
\item If \(\gamma(f_i)<\gamma(f_j)\) and \(a_i,a_j<1\), then  \(\gamma(f_i)\le \gamma(f_j\circ f_i)\le\gamma(f_j)\),
\item If \(\gamma(f_i)<\gamma(f_j)\), \(a_i<1\), \(a_j\ge 1\), and \(a_i\cdot a_j\ge 1\),
then \(\gamma(f_j\circ f_i)\ge\gamma(f_j)~(>\gamma(f_i))\),
\item If \(\gamma(f_i)<\gamma(f_j)\), \(a_i<1\), \(a_j\ge 1\), and \(a_i\cdot a_j< 1\), 
then  \(\gamma(f_j\circ f_i)\le\gamma(f_i)~(<\gamma(f_j))\),
\item If \(\gamma(f_i)<\gamma(f_j)\), \(a_i\ge 1\), \(a_j<1\), and \(a_i\cdot a_j\ge 1\),
then  \(\gamma(f_j\circ f_i)\le\gamma(f_i)~(<\gamma(f_j))\),
\item If \(\gamma(f_i)<\gamma(f_j)\), \(a_i\ge 1\), \(a_j<1\), and \(a_i\cdot a_j< 1\),
then \(\gamma(f_j\circ f_i)\ge\gamma(f_j)~(>\gamma(f_i))\).
\end{enumerate}
% where we allow to write \(+\infty\le +\infty\), \(-\infty\le +\infty\), and 
% \(-\infty\le -\infty\).
\end{lemma:gammacomp}
\begin{proof}
\renewcommand{\theenumi}{$(\roman{enumi})$}
To prove the theorem, we use the following facts for a real \(c\) and a linear function \(f(x)=ax+b\):
\begin{enumerate}
\item If \(a>1\), then \(f(c)> c\siff \gamma(f)<c\), \(f(c)<c \siff \gamma(f)>c\), and \(f(c)=c \siff \gamma(f)=c\).\label{lemma:gammaul+}
\item If \(a<1\), then \(f(c)> c\siff \gamma(f)>c\), \(f(c)<c \siff \gamma(f)<c\), and \(f(c)=c \siff \gamma(f)=c\).\label{lemma:gammaul-}
\item If \(a=1\), then \(f(c)\ge c\siff \gamma(f)=-\infty\), \(f(c)<c \siff \gamma(f)=+\infty\).\label{lemma:gammaul1}
\end{enumerate}
%% \begin{proof}
%% \begin{enumerate}
%% \item It holds because \(f(\gamma(f))=\gamma(f)\) and \(f(x)-x\) is  monotone increasing if \(a>1\).
%% \item It holds because \(f(\gamma(f))=\gamma(f)\) and \(f(x)-x\) is  monotone decreasing if \(a>1\).
%% \item This statement directly follows from the definition of \(\gamma\).
%% \end{enumerate}
%% \end{proof}

\renewcommand{\theenumi}{$(\alph{enumi})$}
\begin{enumerate}
%a
\item Let \(d=\gamma(f_i)=\gamma(f_j)\).
If \(d=+\infty\), then \(a_i=a_j=1\)  and \(b_i,b_j<0\).
Thus, \(\gamma(f_j\circ f_i)=\gamma(x+b_i+b_j)=+\infty\).
If \(d=-\infty\), then \(a_i=a_j=1\)  and \(b_i,b_j\ge 0\).
Thus, \(\gamma(f_j\circ f_i)=\gamma(x+b_i+b_j)=-\infty\).
Otherwise (i.e., \(a_i,a_j\ne 1\)), we have \(f_i(x)=a_i(x-d)+d\) and \(f_j(x)=a_j(x-d)+d\).
Therefore, \(f_j\circ f_i(x)=a_ia_j(x-d)+d\) and \(\gamma(f_j\circ f_i)=d\).

%b
\item By \ref{lemma:gammaul+} and \ref{lemma:gammaul1} and \(\gamma(f_i)<\gamma(f_j)\), we have
\begin{align}
f_j\circ f_i(\gamma(f_i))&= f_j(\gamma(f_i))\le \gamma(f_i),\label{eq:gammacomp++1}\\
f_j\circ f_i(\gamma(f_j))&\ge f_j(\gamma(f_j))= \gamma(f_j).\label{eq:gammacomp++2}
\end{align}
Therefore, we obtain \(\gamma(f_i)\le \gamma(f_j\circ f_i)\le\gamma(f_j)\)
where the first inequality holds by \eqref{eq:gammacomp++1} and by \ref{lemma:gammaul+} and \ref{lemma:gammaul1},
and the second inequality holds by \eqref{eq:gammacomp++2} and by \ref{lemma:gammaul+} and \ref{lemma:gammaul1},

%c
\item By \ref{lemma:gammaul-} and \(\gamma(f_i)<\gamma(f_j)\), we have
\begin{align}
f_j\circ f_i(\gamma(f_i))&= f_j(\gamma(f_i))\ge \gamma(f_i),\label{eq:gammacomp--1}\\
f_j\circ f_i(\gamma(f_j))&\le f_j(\gamma(f_j))= \gamma(f_j).\label{eq:gammacomp--2}
\end{align}
Therefore, we obtain \(\gamma(f_i)\le \gamma(f_j\circ f_i)\le\gamma(f_j)\)
where the first inequality holds by \eqref{eq:gammacomp--1} and by \ref{lemma:gammaul-},
and the second inequality holds by \eqref{eq:gammacomp--2} and by \ref{lemma:gammaul-}.

%d
\item By \ref{lemma:gammaul-} and \(\gamma(f_i)<\gamma(f_j)\), we have
\begin{align*}
f_i\circ f_j(\gamma(f_j))= f_i(\gamma(f_j))\le \gamma(f_j).
\end{align*}
Therefore, we obtain \(\gamma(f_i\circ f_j)\ge \gamma(f_j)\) by \ref{lemma:gammaul+} and \ref{lemma:gammaul1}.

%e
\item By \ref{lemma:gammaul+} and \ref{lemma:gammaul1} and \(\gamma(f_i)<\gamma(f_j)\), we have
\begin{align*}
f_i\circ f_j(\gamma(f_i))\le f_i(\gamma(f_i))= \gamma(f_i).
\end{align*}
Therefore, we obtain \(\gamma(f_i\circ f_j)\le \gamma(f_i)\) by \ref{lemma:gammaul-}.

%f
\item By \ref{lemma:gammaul-} and \(\gamma(f_i)<\gamma(f_j)\), we have
\begin{align*}
f_i\circ f_j(\gamma(f_i))\ge f_i(\gamma(f_i))= \gamma(f_i).
\end{align*}
Therefore, we obtain \(\gamma(f_i\circ f_j)\le \gamma(f_i)\) by \ref{lemma:gammaul+} and \ref{lemma:gammaul1}.

%g
\item By \ref{lemma:gammaul+}, \ref{lemma:gammaul1}, and \(\gamma(f_i)<\gamma(f_j)\), we have
\begin{align*}
f_i\circ f_j(\gamma(f_j))= f_i(\gamma(f_j))\ge \gamma(f_j).
\end{align*}
Therefore, we obtain \(\gamma(f_i\circ f_j)\ge \gamma(f_j)\)
by \ref{lemma:gammaul-}.
\end{enumerate}
\end{proof}

\end{document}